\documentclass[twocolumn,prl,superscriptaddress,showpacs,showkeys]{revtex4-1}

\pdfoutput=1

\def\draft{0}  

\usepackage[margin = 1in]{geometry}

\usepackage{amsmath,amsfonts,amsthm,amssymb}

\usepackage{array}
\usepackage{setspace}

\usepackage{tikz}
\usetikzlibrary{math,calc,arrows.meta}
\usetikzlibrary{shapes.geometric}

\usepackage{multirow}

\usepackage[colorlinks = true]{hyperref}
\definecolor{darkred}  {rgb}{0.5,0,0}
\definecolor{darkblue} {rgb}{0,0,0.5}
\definecolor{darkgreen}{rgb}{0,0.5,0}
\hypersetup{
  urlcolor   = blue,         
  linkcolor  = darkblue,     
  citecolor  = darkgreen,    
  filecolor  = darkred       
}

\usepackage{cleveref}
\crefname{lemma}{Lemma}{Lemmas}
\crefname{definition}{Definition}{Definitions}
\crefname{theorem}{Theorem}{Theorems}
\crefname{section}{Section}{Sections}
\crefname{appendix}{Appendix}{Appendices}
\crefname{figure}{Fig.}{Figs.}
\crefname{equation}{Eq.}{Eqs.}
\crefname{table}{Table}{Tables}
\crefname{example}{Example}{Examples}

\newtheorem{theorem}{Theorem}
\newtheorem*{lemma*}{Lemma}
\newtheorem{lemma}{Lemma}
\newtheorem*{claim}{Hoeffding's inequality}

\newcommand{\ket}[1]{|#1\rangle}
\newcommand{\bra}[1]{\langle#1|}

\newcommand{\proj}[1]{|#1\rangle\langle#1|}

\DeclareMathOperator{\Tr}{Tr}
\DeclareMathOperator{\poly}{poly}
\DeclareMathOperator{\cc}{cc}

\newcommand{\C}{\mathbb{C}}

\newcommand{\A}{\mathcal{A}}

\DeclareFontFamily{U}{mathc}{}
\DeclareFontShape{U}{mathc}{m}{it}{<->s*[1.03] mathc10}{}
\DeclareMathAlphabet{\mathscr}{U}{mathc}{m}{it}
\renewcommand{\a}{\mathscr{a}}

\newcommand{\E}{\mathbb{E}}

\newcommand{\mx}[1]{\begin{pmatrix}#1\end{pmatrix}}
\newcommand{\smx}[1]{\bigl(\begin{smallmatrix}#1\end{smallmatrix}\bigr)}

\usepackage{mathtools}
\mathtoolsset{centercolon}
\DeclarePairedDelimiter{\set}{\lbrace}{\rbrace}
\DeclarePairedDelimiter{\abs}{\lvert}{\rvert}
\DeclarePairedDelimiter{\norm}{\lVert}{\rVert}
\DeclarePairedDelimiter{\of}{\lparen}{\rparen}
\DeclarePairedDelimiter{\sof}{\lbrack}{\rbrack}

\newcommand{\ct}{^{\dagger}}

\newcommand{\x}{\otimes}

\newcommand{\xp}[1]{^{\otimes #1}}

\newcommand{\Part}{\mathcal{P}}
\newcommand{\Ecut}{{E_{\text{cut}}}}

\newsavebox{\Hgate}
\savebox{\Hgate}{$H:=\smx{1&1\\1&-1}/\sqrt{2}$}
\newsavebox{\Sgate}
\savebox{\Sgate}{$S:=\smx{1&0\\0&i}$}

\newcommand{\CR}[1]{#1} 

\newcommand{\Cb}[1]{#1}

\ifnum \draft=1
\newcommand{\Lone}[1]{{\color{red} \textbf{#1}}}
\newcommand{\Ltwo}[1]{{\color{blue} #1}}
\else
\newcommand{\Lone}[1]{#1}
\newcommand{\Ltwo}[1]{}
\fi

\newcommand{\sect}[1]{\emph{#1.---}}

\begin{document}
\title{Simulating Large Quantum Circuits on a Small Quantum Computer}

\author{Tianyi Peng}
\email{tianyi@mit.edu}
\affiliation{Laboratory for Information and Decision Systems,
Massachusetts Institute of Technology, Cambridge, MA 02139, USA.}

\author{Aram W.~Harrow}
\affiliation{Center for Theoretical Physics,
Massachusetts Institute of Technology, Cambridge, MA 02139, USA.}

\author{Maris Ozols}
\affiliation{University of Amsterdam and QuSoft, Amsterdam, the Netherlands.}

\author{Xiaodi Wu}
\email[corresponding author: ]{xwu@cs.umd.edu}
\affiliation{Department of Computer Science, Institute for Advanced Computer Studies, and Joint Center for Quantum Information and Computer Science, University of Maryland, College Park, MD 20742, USA.}


\begin{abstract}
Limited quantum memory is one of the most important constraints for near-term quantum devices. Understanding whether a small quantum computer can simulate a larger quantum system, or execute an algorithm requiring more qubits than available, is both of theoretical and practical importance.
In this Letter, we introduce cluster parameters $K$ and $d$ of a quantum circuit. The tensor network of such a circuit can be decomposed into clusters of size at most $d$ with at most $K$ qubits of inter-cluster quantum communication.
\Cb{We propose a \textit{cluster simulation scheme} that can simulate any $(K,d)$-clustered quantum circuit on a $d$-qubit machine in time roughly $2^{O(K)}$, with further speedups possible when taking more fine-grained circuit structure into account.
We show how our scheme can be used to simulate clustered quantum systems---such as large molecules---that can be partitioned into multiple significantly smaller clusters with weak interactions among them.
By using a suitable clustered ansatz, we also experimentally demonstrate that a quantum variational eigensolver can still achieve the desired performance for estimating the energy of the BeH$_2$ molecule while running on a physical quantum device with half the number of required qubits.}
\end{abstract}



\maketitle

\sect{\Lone{Introduction}}%
Near-term quantum computing applications will focus on Noisy Intermediate-Scale Quantum (NISQ) devices~\cite{preskill2018quantum}, where quantum memory is limited both in quantity and quality.
To meet the memory need of such applications (e.g., quantum simulation~\cite{lloyd1996universal, cirac2012goals, o2016scalable}, quantum optimization~\cite{farhi2014quantum, moll2017quantum}, and quantum machine learning~\cite{rebentrost2014quantum, biamonte2017quantum}), it is desirable to seek a way to perform computations that require more qubits than physically available, at the cost of additional affordable classical processing.

Trading classical computation for quantum computation is a well-motivated topic of long-standing interest. An extreme example of this is
the \emph{(fully) classical simulation} with no quantum computation at all, which is however limited to small dimensions, weak entanglement, or circuits with special gate sets~\cite{MS08, Vidal03, chen2018classical, SDV06, Jozsa}.
Recently, the possibility of trading classical computation for  ``virtual qubits''
has been discussed in~\cite{BSS16}.
A systematic understanding of such trade-offs will be crucial for realizing near-term quantum applications.

In this Letter, we \Cb{introduce a \textit{cluster simulation scheme}}, a general framework for simulating large quantum circuits on a quantum computer with a small amount of quantum memory. The performance of our simulation depends on the \textit{cluster parameters} of the given circuits.
In particular, we are inspired by the classical \emph{fragmentation} methods and \emph{Quantum Mechanics/Molecular Mechanics} (QM/MM) methods for simulating molecules~\cite{gordon2011fragmentation, li2007generalized, warshel1976theoretical, li2008fragmentation}
that can be partitioned into multiple weakly interacting clusters of significantly smaller size
(e.g., an oligosaccharide consisting of several monosaccharides).
Following the spirit of~\cite{BSS16} and~\cite{gordon2011fragmentation, li2007generalized, warshel1976theoretical, li2008fragmentation}, a natural definition of cluster parameters of a quantum circuit should capture the decomposability of the circuit into clusters of bounded size and limited inter-cluster interactions.

Our definition of the cluster parameters is guided by the above intuition, but with an important distinction.
Instead of looking into the decomposability of any given circuit, we are concerned about the decomposability of the corresponding \emph{tensor network}, which is inspired by the tensor-network-based classical simulation of quantum circuits~\cite{MS08, arfken1999mathematical, orus2014practical, AL10}.
One significant advantage of our definition, as we will see below, is to use more flexible decompositions of tensor networks than are possible with simple partitioning of qubits (e.g., as in~\cite{BSS16}).

\begin{figure*}
\begin{tikzpicture}[semithick, > = latex,
	gate/.style = {fill = white, draw, inner sep = 6pt},
		pt/.style = {draw, fill = white, inner sep = 3pt},
		tr/.style = {isosceles triangle, inner sep = 0pt, minimum width = 8pt, isosceles triangle apex angle = 70},
  	me/.style = {tr, draw, fill = white},
		cr/.style = {tr, draw, fill = white, shape border rotate = 180},
		bd/.style = {draw = none, fill = black!15, rounded corners = 2pt},
		lo/.style = {orange!20},
		lb/.style = {blue!20}
	]


\def\H{18pt}
\def\W{23pt}
\def\arrowscale{1}

\newcommand{\gate}[3]{
	\draw[gate] (#1*\W-0.45*\H,#2*\H-0.5*#3*\H+0.15*\H)
		rectangle (#1*\W+0.45*\H,#2*\H+0.5*#3*\H-0.15*\H);
}

\newcommand{\measure}[2]{
	\def\r{0.22}
	\gate{#1}{#2}{1}
	\draw[-{Latex[scale=0.8*\arrowscale]}] (#1*\W,#2*\H-0.2*\H) -- (#1*\W+0.28*\W,#2*\H+0.3*\H);
	\draw (#1*\W+\r*\W,#2*\H-0.2*\H) arc [start angle = 0, end angle = 180, radius = \r*\W];
}

\newcommand{\gnode}[3]{
  \node[pt] (#3) at (#1*\W,#2*\H) {};
}

\newcommand{\circuit}[1]{
	\foreach \y in {0,1,2,3} {
	  \node at (0,\y*\H) {$\ket{0}$};
	}
	\foreach \y in {0,2,3} {
  	\draw (0.3*\W,\y*\H) -- (3*\W+#1*\W,\y*\H);
	}
  \draw (0.3*\W,\H) -- (1*\W,\H);
	\draw (2*\W+#1*\W,\H) -- (3*\W+#1*\W,\H);
	\gate{1}{3.0}{1}
	\gate{1}{1.5}{2}
	\begin{scope}[xshift = #1*\W]
  	\gate{2}{2.5}{2}
	  \gate{2}{0.5}{2}
	  \measure{3.0}{0}
	  \measure{3.0}{1}
  	\measure{3.0}{3}
		\measure{3.0}{2}
  \end{scope}
}

\tikzset{modified circuit/.pic = {

\def \d{1pt}
\fill[bd,lb]
     (-0.4*\W,    0.5*\H+\d)
  -- ( 2.5*\W-\d, 0.5*\H+\d)
	-- ( 2.5*\W-\d, 1.5*\H+\d)
	-- ( 5.5*\W,    1.5*\H+\d)
	-- ( 5.5*\W,    3.5*\H)
	-- (-0.4*\W,    3.5*\H) -- cycle;

\fill[bd,lo]
	   (-0.4*\W,    0.5*\H-\d)
	-- ( 2.5*\W+0.5pt, 0.5*\H-\d)
	-- ( 2.5*\W+0.5pt, 1.5*\H-\d)
	-- ( 5.5*\W,    1.5*\H-\d)
	-- ( 5.5*\W,   -0.5*\H)
	-- (-0.4*\W,   -0.5*\H) -- cycle;

  \def\s{2}
  \draw[red] (\W,\H) -- (2*\W,\H);
  \draw[red] (1.4*\W+\s*\W,\H) -- (2*\W+\s*\W,\H);
	\circuit{\s}
	\measure{2.0}{1}
	\node at (2.1*\W, \H+0.67*\H) {\footnotesize $O_{#1}$};
	\node at (\W+\s*\W,\H) {\small$\ket{\rho_{#1}}$};
}}


\tikzset{phase1/.pic = {
\begin{scope}[scale = #1, xshift = -2*\W, yshift = -1.5*\H]
	\def\d{1pt}
	\fill[bd,lb]
     (-0.3*\W,    0.5*\H+\d)
  -- ( 1.5*\W-\d, 0.5*\H+\d)
	-- ( 1.5*\W-\d, 1.5*\H+\d)
	-- ( 3.5*\W,    1.5*\H+\d)
	-- ( 3.5*\W,    3.5*\H)
	-- (-0.3*\W,    3.5*\H) -- cycle;
	\fill[bd,lo]
	   (-0.3*\W,    0.5*\H-\d)
	-- ( 1.5*\W+\d, 0.5*\H-\d)
	-- ( 1.5*\W+\d, 1.5*\H-\d)
	-- ( 3.5*\W,    1.5*\H-\d)
	-- ( 3.5*\W,   -0.5*\H)
	-- (-0.3*\W,   -0.5*\H) -- cycle;
	\draw[red] (\W,\H) -- (2*\W,\H);
	\circuit{0}
	\draw[rounded corners = 5pt] (3.7*\W+0.21*\W,1.5*\H) -- +(-5pt,0) -- (3.7*\W, 3.3*\H) -- +(-5pt,0);
	\draw[rounded corners = 5pt] (3.7*\W+0.21*\W,1.5*\H) -- +(-5pt,0) -- (3.7*\W,-0.3*\H) -- +(-5pt,0);
	\node (y) at (3.8*\W+0.5*\W,1.5*\H) {$f(y)$};

	\node at (2*\W, 4.5*\H) {$(C, f)$};

\end{scope}
}}


\tikzset{network/.pic = {

	\foreach \y in {0,1,2,3} {
    \node[cr] (\y) at (0,\y*\H) {};
  }

	\gnode{1}{3.0}{A}
  \gnode{1}{1.5}{B}
  \gnode{2}{2.5}{C}
  \gnode{2}{0.5}{D}

 	\node[me] (F) at (3.17*\W,1.5*\H) {};

	\draw[->] (3) -- (A);
  \draw[->] (2) -- (B);
	\draw[->] (1) -- (B);
	\draw[->] (0) -- (D);

  \draw[->] (A) -- (C);
  \draw[->] (B) -- (C);

	\draw[->] ($(C)+(0, 1.5pt)$) -- ($(F.west)+(0,3pt)$);
	\draw[->] ($(C)+(0,-1.5pt)$) -- (F.west);

	\draw[->] ($(D)+(0, 1.5pt)$) -- (F.west);
	\draw[->] ($(D)+(0,-1.5pt)$) -- ($(F.west)+(0,-3pt)$);

	\gnode{2}{2.5}{C}
  \gnode{2}{0.5}{D}

}}


\tikzset{phase2/.pic = {
\begin{scope}[scale = #1, xshift = -1.5*\W, yshift = -1.5*\H]

  \fill[bd,lb] (-0.3*\W,0.5*\H) -- (2.3*\W,1.7*\H) -- (2.3*\W, 3.5*\H) -- (-0.3*\W, 3.5*\H) -- cycle;
	\fill[bd,lo] (-0.3*\W,0.2*\H) -- (2.3*\W,1.4*\H) -- (2.3*\W,-0.5*\H) -- (-0.3*\W,-0.5*\H) -- cycle;
	\node at (1*\W, 2.3*\H) {$S_1$};
  \node at (2*\W,-0.1*\H) {$S_2$};

	\pic at (0,0) {network};
	\node at (F) [label = $O_f$] {};
	\draw[->,red] (B) -- (D);
	\node[red] at (1.9*\W,1.5*\H) {$K=1$};
	\node      at (3.0*\W,3.2*\H) {$d=3$};

  \tikzset{txt/.style = {baseline, anchor = west}}
  \path (0*\W,-1.2*\H)
  	        node[cr] {} +(5pt,1pt) node[txt] {state}
++(0,-12pt) node[pt] {} +(5pt,1pt) node[txt] {gate}
++(0,-12pt) node[me] {} +(5pt,1pt) node[txt] {observable};

\node at (1.6*\W, 4.5*\H) {$(G, \mathcal{A})$};

\end{scope}
}}


\tikzset{modified network/.pic = {
\begin{scope}[scale = 0.9, every node/.style = {scale = 0.9}]

\fill[bd,lb] (-0.3*\W,0.5*\H)
	-- (1*\W,1.1*\H)
	-- (2.3*\W, 1.1*\H)
	-- (2.3*\W, 3.25*\H)
	-- (1*\W, 3.25*\H)
	-- (-0.3*\W, 3.25*\H) -- cycle;

	\fill[bd,lo] (-0.3*\W,0.3*\H)
	-- (1*\W,0.9*\H)
	-- (2.3*\W,  0.9*\H)
	-- (2.3*\W, 0.2*\H)
	-- (-0.3*\W,-0.5*\H) -- cycle;

  \pic at (0,0) {network};
  \node[me] (M) at ($(B)+( \W,-0.1*\H)$) {};
  \node[cr] (S) at ($(D)+(-\W, 0.1*\H)$) {};
	\draw[->,red] (B) -- (M);
	\draw[->,red] (S) -- (D);

	\node at (2.1*\W, 1.9*\H) {\footnotesize $O_{#1}$};

	\node at (0.6*\W, 0.5*\H) {\footnotesize $\rho_{#1}$};

\end{scope}
}}


\tikzset{phase3/.pic = {
\begin{scope}[scale = #1, xshift = -1cm, yshift = -1.2cm]

  \pic at (0, 1.9) {modified network = {1}};
	\pic at (0,-1.3) {modified network = {8}};

	\node at (1.0,1.7) {$+$};
	\node at (1.0,1.3) {$\vdots$};
	\node at (1.0,0.7) {$+$};

	\node at (2.35, 3.6) {\small $(G', \mathcal{A}_1)$};

	\node at (2.35, 0.4) {\small $(G', \mathcal{A}_8)$};

\end{scope}
}}


\tikzset{phase4/.pic = {
\begin{scope}[scale = #1, xshift = -1cm, yshift = -1.2cm]

  \def\arrowscale{0.7}
  \pic[scale = 0.75] at (-0.2, 2.1) {modified circuit = {1}};
	\pic[scale = 0.75] at (-0.2,-1.2) {modified circuit = {8}};

	\node at (1.3,1.65) {$+$};
	\node at (1.3,1.25) {$\vdots$};
	\node at (1.3,0.65) {$+$};

\end{scope}
}}


\pic at ( 0.0,2) {phase1 = {1.0}}; \node at ( 0.0, 5) {Phase 1}; \draw[dashed] ( 2.5,-0.7) -- +(0,5.2);
\pic at ( 4.4,2) {phase2 = {1.0}}; \node at ( 4.4, 5) {Phase 2}; \draw[dashed] ( 6.4,-0.7) -- +(0,5.2);
\pic at ( 8.0,2) {phase3 = {1.0}}; \node at ( 8.0, 5) {Phase 3}; \draw[dashed] (10.0,-0.7) -- +(0,5.2);
\pic at (12.0,2) {phase4 = {1.0}}; \node at (12.0, 5) {Phase 4};

\end{tikzpicture}
\caption{Four phases of our simulation: (1)~the original quantum circuit; (2)~the corresponding tensor network; (3)~a collection of tensor networks obtained by cutting an edge; (4) a collection of smaller quantum circuits.}
\label{fig:overall}
\end{figure*}
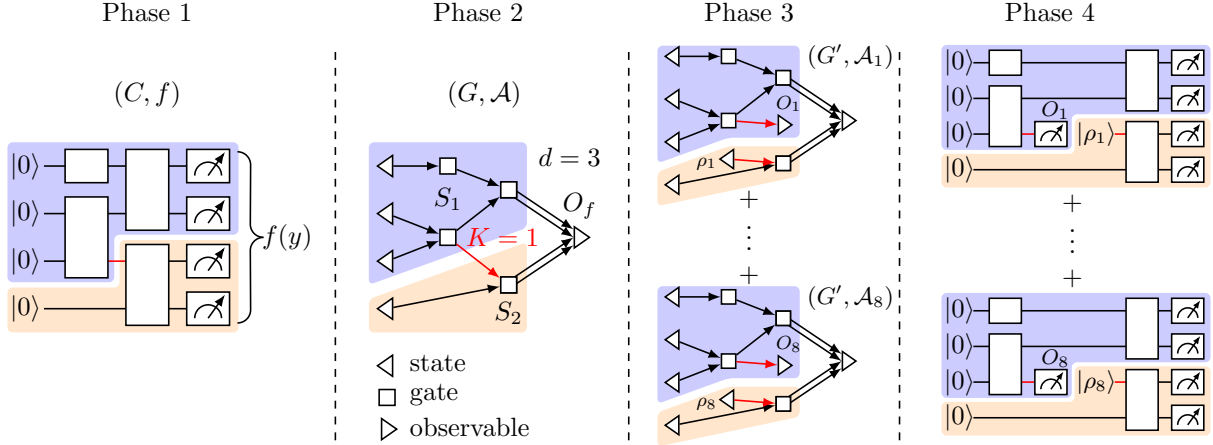

Given our definition of clustered circuits, our main contribution is a scheme to simulate the entire quantum circuit by simulating each cluster on a small quantum machine with classical post-processing. \CR{The key difference between our scheme and fully classical simulation schemes  \cite{bernstein1997quantum, MS08, Vidal03, chen2018classical, SDV06, Jozsa, markov2018quantum, boixo2017simulation, aaronson2016complexity, chen201864} is that we keep part of the computation quantum (i.e., unitary). In particular}, we design a method to simulate inter-cluster quantum interactions by classical means. Comparing to ``virtual qubits'' in~\cite{BSS16}, our technique can be deemed as trading classical computation for ``virtual quantum communication''.
The cluster simulation scheme applies to general quantum circuits, which distinguishes it from application-specific techniques for saving qubits~\cite{romero2017quantum, bravyi2017tapering, moll2016optimizing, steudtner2017lowering}.

We apply our scheme to Hamiltonian simulation~\cite{vidal2004efficient, berry2007efficient, berry2015simulating, low2017optimal}, particularly for clustered quantum systems, and Variational Quantum Eigensolvers (VQE), a popular candidate for showing quantum advantages on near-term quantum devices~\cite{mcclean2016theory, barrett2013simulating, wecker2015progress}.
In both applications, we show that the particular quantum circuits have favorable cluster parameters that are amenable to our techniques.
One can also interpret our technique as a \emph{hybrid variational ansatz} in which a quantum computer is used for some parts of the circuit and a classical computer is used for other parts, which might be of independent interest.

Our scheme can easily be extended to allow limited inter-cluster quantum communication. It can hence be leveraged to perform general quantum circuits on modular quantum systems, a leading proposal of scalable quantum computers (e.g.,~\cite{MSL16, PhysRevA.76.062323, Yao12, Dai20}).

\sect{\Lone{Computational model}}
We use the same computational model as in~\cite{BSS16}; see \cref{fig:overall} (Phase~1).
An $m$-gate quantum circuit $C$ with 1- and 2-qubit gates is applied to $\ket{0}\xp{n}$ and all output qubits are measured in the computational basis.
A classical post-processing function $f: \set{0,1}^n \to [-1,1]$ is then applied to the measurement outcomes.
We assume that $f$ can be efficiently computed classically.
We call the overall procedure a \emph{quantum-classical algorithm} (QC algorithm) and denote it by $(C,f)$. Its expected output $\E_y f(y)$ is averaged over all measurement outcomes $y \in \set{0,1}^n$.
The goal of our simulation is to approximate $\E_y f(y)$ within precision $\epsilon$ with high probability, say  at least $2/3$.

\sect{\Lone{Clustered circuits}}%
Any QC algorithm $(C,f)$ can be represented by a tensor network $(G,\A)$ consisting of a directed graph $G(E,V)$ and a collection of tensors $\A=\{A(v): v \in V\}$.
The vertices $V$ of $G$ represent individual gates (denoted by~$\Box$), input qubits (denoted by~$\lhd$), and observables (denoted by~$\rhd$) as shown in \cref{fig:overall} (Phase~2), whereas the flow of qubits is encoded by the directed edges $E$ of $G$.
Note that each gate vertex~$\Box$ has the same in- and out-degree
(i.e., the same number of incoming and outgoing edges)
whereas~$\lhd$ vertices only have outgoing edges and~$\rhd$ vertices only have incoming edges.
For each $v \in V$, $A(v)$ is a tensor that encodes the matrix entries of the corresponding gate, state, or observable, and the value $T(G,\A)$ of the tensor network $(G,\A)$ coincides with the output expectation of the corresponding $(C,f)$ algorithm, i.e.,
\begin{equation}
  T(G,\A) = \E_y f(y).
\end{equation}
See the Supplemental Material for more details.

A QC algorithm $(C,f)$  is \emph{$(K,d)$-clustered} if its tensor network $(G,\A)$ has the following structure.
Setting the final observable $O_f$ aside, we partition the remaining vertices of $G$ into clusters $S_1, \dotsc, S_r$ and let $g$ be the $(r+1)$-vertex multi-graph obtained by contracting each cluster to a single vertex.
Let $K$ be the number of edges in $g$ minus the in-degree of $O_f$ (intuitively, $K$ is the total number of qubits communicated between clusters)
and let $d$ be the number of qubits sufficient for simulating each cluster.

While finding the minimal $d$ can be non-trivial (especially if qubits can be recycled after measurement),
a good estimate of $d$ is $\max_{i}d(S_i)$ where $d(S_i)$ is the out-degree of cluster $S_i$.
This is a valid upper bound on the minimal $d$ since $d(S_i)$ is the number of~$\lhd$ vertices in $S_i$ plus the number of incoming edges to $S_i$, which upper bounds the total number of qubits required to simulate $S_i$.

For example, in \cref{fig:overall} (Phase~2), two parts of a partition $\{S_1, S_2\}$ are indicated by blue and orange, respectively. Since only one qubit is sent from $S_1$ to $S_2$, $K=1$.
We have $d(S_1) = 3$ due to two outgoing edges from $S_1$ to $O_f$ and one from $S_1$ to $S_2$. Similarly, $d(S_2)=2$ and thus we can take $d=3$. This circuit is hence $(1,3)$-clustered.

Our framework generally allows for more flexibility in decomposing quantum circuits compared to~\cite{BSS16}. Consider the $2n$-qubit example in \cref{fig:tensornet}. Assuming each block $B_i$ with \CR{depth $D$} is dense, i.e., contains two-qubit gates between all pairs of qubits,
any partition of the initial $2n$ qubits induces at least $\Omega(n)$ gates across the parties and thus requires $\Omega(n)$ qubits of communication to implement the circuit.
However, this is a $(2, n+1)$-clustered circuit with only two qubits of communication between the blue and orange clusters in \cref{fig:tensornet}. \CR{Furthermore, the depth is reduced from $3D$ to $2D$ when simulating each cluster separately.}

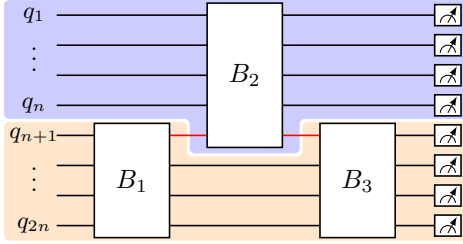
\begin{figure}
\begin{tikzpicture}[semithick, > = latex,
	gate/.style = {fill = white, draw, inner sep = 6pt},
		pt/.style = {draw, fill = white, inner sep = 3pt},
		tr/.style = {isosceles triangle, inner sep = 0pt, minimum width = 8pt, isosceles triangle apex angle = 70},
  	me/.style = {tr, draw, fill = white},
		cr/.style = {tr, draw, fill = white, shape border rotate = 180},
		bd/.style = {draw = none, rounded corners = 2pt},
		lo/.style = {orange!20},
		lb/.style = {blue!20}
	]

\newcommand{\meas}[2]{
  \def\r{0.1}
  \def\AR{0.12}

  \def \HM{0.35}
  \draw[gate] (#1-0.5*\HM,#2-0.5*\HM+0.1*\HM)
    rectangle (#1+0.5*\HM,#2+0.5*\HM-0.1*\HM);

  \draw[very thin,->] (#1,#2-\r/2-0.05) -- (#1+\AR,#2+\AR);
  \draw[thin] (#1+\r,#2-\r/2-0.05) arc [start angle = 0, end angle = 180, radius = \r];
}

  \def\H{0.4}
  \def\W{1}
  \def\G{8}

\def\d{0.02}

\fill[bd,lb]
     (-0.7*\W,    4.5*\H+\d)
     -- (1.75*\W+\d, 4.5*\H+\d)
     -- (1.75*\W+\d, 3.35*\H+\d)
     -- (3.25*\W-\d, 3.35*\H+\d)
     -- (3.25*\W-\d, 4.5*\H+\d)
  -- ( 5.5*\W, 4.5*\H+\d)
	-- ( 5.5*\W, 8.5*\H)
	-- ( -0.7*\W,    8.5*\H) -- cycle;

\fill[bd,lo]
     (-0.7*\W,    0.5*\H)
  -- ( 5.5*\W, 0.5*\H)
	-- ( 5.5*\W, 4.5*\H-\d)
	-- ( 3.25*\W+\d,    4.5*\H-\d)
	-- ( 3.25*\W+\d,    3.35*\H-\d)
	-- (1.75*\W-\d, 	3.35*\H-\d)
	--(1.75*\W-\d, 4.5*\H-\d)
	-- (-0.7*\W,    4.5*\H-\d) -- cycle;

  \foreach \x in {1,2,3,5,6,7,\G} {
    \draw (0,\x*\H) -- (5.2*\W,\x*\H);
    \meas{5.2*\W}{\x*\H}
  }

 \draw (0, 4*\H) -- (1*\W, 4*\H);

  \draw[red] (1*\W, 4*\H) -- (2*\W, 4*\H);

  \draw[red] (3*\W, 4*\H) -- (4*\W, 4*\H);

 \draw (4, 4*\H) -- (5.2*\W, 4*\H);
 \meas{5.2*\W}{4*\H};

  \node at (-0.3,\H) {\small$q_{2n}$};
  \node at (-0.3,4*\H) {\small$q_{n+1}$};
  \node at (-0.3,5*\H) {\small$q_{n}$};
  \node at (-0.3,8*\H) {\small$q_{1}$};
  \node at (-0.3, 2.8*\H) {\small$\vdots$};
  \node at (-0.3, 6.8*\H) {\small $\vdots$};

  \def\w{0.8}

  \draw[gate] (0.5*\W,0.6*\H) rectangle (1.5*\W,4.4*\H);
  \draw[gate] (2*\W,3.6*\H) rectangle (3*\W,\G.4*\H);
  \draw[gate] (3.5*\W,0.6*\H) rectangle (4.5*\W,4.4*\H);

  \node at (1*\W, 2.5*\H) {$B_1$};
  \node at (2.5*\W, 6*\H) {$B_2$};
  \node at (4*\W, 2.5*\H) {$B_3$};

\end{tikzpicture}
\caption{A $(2, n+1)$-clustered circuit with three dense blocks.
While any partition of its qubits induces $\Omega(n)$ gates between different parts,
merging blocks $B_1$ and $B_3$ into a single cluster results in only two qubits communicated between the two clusters. \CR{In this example, the size and the depth of the circuit after clustering are both reduced compared to the original circuit.}}
\label{fig:tensornet}
\end{figure}

\sect{\Lone{Cluster simulation scheme}}%
To fulfill the above intuition, for any $(K,d)$-clustered circuit, we need to show how (i) each cluster can be simulated on a $d$-qubit quantum machine and (ii) how to simulate the interaction among clusters.
We design an edge-cutting procedure to decompose tensor networks as shown in \cref{fig:overall} (Phase~3).
In particular, we replace each edge, modeled as a perfect channel for communicating a qubit, by a collection of tensor networks that reproduce the inter-cluster communication by operations within each cluster.
This, however, comes at a cost of having to average over several runs.
(In the spirit of~\cite{BSS16}, this technique can be thought of as ``virtual quantum communication.'')

\begin{lemma}\label{lem:edgecut}
Let $(G(E,V),\A)$ be a tensor network of a QC algorithm. For any edge $e \in E$,
\begin{equation} \label{eqn:edge}
  T(G,\A) = \sum_{i=1}^8 c_i T(G',\A_i),
\end{equation}
where $G'$ differs from $G$ by removing $e$ and adding one $\lhd$ and one $\rhd$ vertex, each $c_i \in \set{-\frac{1}{2},\frac{1}{2}}$, and each $(G',\A_i)$ corresponds to a valid quantum circuit.
\end{lemma}
(All proofs in this Letter are deferred to the Supplemental Material~\footnote{See Supplemental Material for more details about (1) the precise definition of tensor network corresponding to clustered circuits; (2) the proofs of Lemma 1, Theorem 1, 2, 3; (3) the details of experiments about VQE, which includes Refs.~\cite{Hoeffding,miller1973symmetry,mitarai2019constructing,RevModPhys.82.277}.}.) By repeating this process and deleting more edges, the tensor network can eventually be partitioned into individual clusters. Each cluster will only have outgoing edges to $O_f$ and can hence be simulated by a $d$-qubit quantum computer plus classical processing of the measurement outcomes (Phase 4 in \cref{fig:overall}).
We combine individual simulation results by a simple sampling procedure according to Eq. (\ref{eqn:edge}).

Our overall simulation scheme consists of several iterations of the following steps: (i)~producing a classical description of a quantum circuit with $O(m)$ gates and $d$ qubits (potentially recycled during the circuit), (ii)~running this circuit on $\ket{0}\xp{d}$, and (iii)~classically post-processing the measurement outputs.
The final step has to produce
with probability at least 2/3
an $\epsilon$-approximation of $T(G,\A)$.

The complexity of our scheme scales with the cluster parameters $(K,d)$ as well as the total number of qubits $n$ and gates $m$ in the original circuit $C$, and the desired additive simulation accuracy $\epsilon$.
The total classical and quantum running time of our simulator is $O(Q\poly(n+m)/\epsilon^2)$, for some exponentially scaling parameter $Q$.
For simplicity, we ignore the polynomial part of the run-time and call this a \emph{$(Q,d)$-simulator}. A fully classical simulator is thus a $(2^{O(n)},0)$-simulator,
while a scalable quantum computer is a $(1,n)$-simulator with an exponentially improved total run-time.
Our result can be deemed as a smooth trade-off between these two extreme cases.

\begin{theorem}\label{thm:main}
Any QC algorithm $(C,f)$ with a $(K,d)$-clustered circuit $C$ has a $\of{2^{O(K)},d}$-simulator.
The total classical and quantum running time of this simulator is $O(2^{4K}(n+m)/\epsilon^2)$, where $n$ and $m$ are the total number of qubits and gates in $C$, and $\epsilon$ is the desired accuracy.
\end{theorem}
\CR{In the special case when there are only two clusters, the number of qubits
  $K$ communicated among the clusters can be regarded as an upper bound on
  entanglement. Hence, the result relates the classical computation cost to the
  entanglement between the two clusters.}

The efficiency of our simulation can be further improved for special classes of post-processing functions $f: \set{0,1}^n \to [-1,1]$. For example, consider \emph{decomposable} $f$ satisfying $f(y) = \prod_{j=1}^r f_j(y_j)$, where $y = y_1 \dots y_r$ is a partition of the original $n$-bit string $y$ into substrings $y_j$ that correspond to outputs of different clusters~\footnote{We have assumed that the number of terms in the decomposition of $f$ agrees with the number of clusters $r$. If some cluster does not produce a qubit that feeds directly into the final observable, we can insert a fictitious function $f_j$ in the decomposition ($f_j$ has no arguments and is identically equal to $1$).}, and $f_j(y_j) \in [-1,1]$.
Typical examples of such decomposable functions arise from Pauli observables in VQE~\cite{kandala2017hardware} or estimating probabilities of specific output strings~\cite{pashayan2017estimation}.
For such functions, we can replace $O_f$ by smaller tensors $O_{f_j}$ and include them in the corresponding clusters $S_j$.
As a result, the induced graph $g$ no longer contains $O_f$.
Nevertheless, we can still apply Lemma~\ref{lem:edgecut} to decompose each cluster and simulate it on a $d$-qubit quantum machine.
However, inspired by~\cite{MS08}, a more efficient scheme for combining individual simulations is now possible.
Its complexity depends on $\cc(g)$---the \emph{contraction complexity} of $g$---that is the minimum (over all possible contraction orders) of the maximum node degree during the procedure of contracting the graph to a single vertex.

\begin{theorem}\label{thm:Treewidth}
Any QC algorithm $(C,f)$ with a $(K,d)$-clustered circuit $C$, a decomposable function $f$, and induced graph $g$ has a $\of{2^{O(\cc(g))},d}$-simulator.
\end{theorem}

Note that $\cc(g) \le K$, where $K$ is the number of edges in $g$.
Compared to $2^{O(K)}$ in \cref{thm:main}, the factor $2^{O(\cc(g))}$ in \cref{thm:Treewidth} is a significant improvement for some families of graphs. For example, among constant-degree graphs with $n$ nodes, $\cc(g)=O(1)$ for trees and $\cc(g) = O(\sqrt{n})$ for planar graphs~\cite{bodlaender1994tourist}, while $K$ can be as large as $O(n)$.

\sect{\Lone{Application to Hamiltonian simulation}}%
One of the most promising potential applications of our result is the simulation of clustered quantum systems.
Specifically, we consider quantum systems with geometric layouts where each qubit only interacts with $O(1)$ adjacent qubits.
The corresponding interaction graph $G$ (i.e., qubits as vertices and interactions as edges) has constant degree.
Assume further that qubits in $G$ can be grouped into $n$ parties and let $g$ be the induced graph obtained by contracting each party of $G$ to a single vertex.
The Hamiltonian of such a system can be written as a sum of local terms
\begin{equation}
  H = \sum_j H^{(1)}_j + \sum_j H^{(2)}_j, \quad
  \forall i,j: \norm{H_j^{(i)}} \leq 1,
  \label{eq:H}
\end{equation}
where each term acts on at most two qubits and the superscripts $(1)$ and $(2)$
indicate that these qubits belong to a single party or two different parties, respectively.
The \emph{interaction strength} between all parties can be characterized by $h = \sum_{j} \norm{H_j^{(2)}}$.
We are interested in quantum systems with weak interaction strength (e.g., \cref{fig:Hamiltonian}).
Assume that the system is initialized in a product state $\rho = \rho_1 \x \dotsb \x \rho_n$, where $\rho_i$ is an efficiently preparable state of the $i$-th party.
Our goal is to approximate the following \emph{correlation function}: $\Tr \sof[\big]{ (O_1 \x \dotsb \x O_n) e^{-iHt} (\rho_1 \x \dotsb \x \rho_n) e^{iHt} },$
where $t$ is the evolution time and $O_i$ is an efficiently measurable observable of the $i$-th part with eigenvalues in $[-1,1]$.

\begin{theorem}\label{thm:Hamiltonian}
The correlation function of the Hamiltonian $H$ in \cref{eq:H}
can be approximated to accuracy $\epsilon$
by a $\of[\big]{2^{O\of*{(ht)^2\cc(g)/\epsilon}}, d}$-simulator, where
$\cc(g)$ is the contraction complexity of its induced graph $g$,
$h$~is the interaction strength, $t$ the evolution time, and $d$ the number of qubits in the largest party.
\end{theorem}
At a high level, the above result is obtained by applying Theorem~\ref{thm:Treewidth} to Hamiltonian simulation circuits of $e^{-iHt}$ based on the Trotter-Suzuki approximation, but with the following important improvements. To obtain a better estimate of the cluster parameters $(K,d)$, we need to apply Lemma~\ref{lem:edgecut} to trim the tensor network beyond simulating inter-cluster communication, and to conduct a careful analysis of $d$ to allow recycling of qubits. Inspired by~\cite{haah2018quantum}, we also need to improve the naive error analysis  and to obtain an error bound in terms of the interaction strength $h$. 

The exponential dependence on $t$ \footnote{This dependence on $t$ has been subsequently improved by using a $p$-th order product formula ($p>1$). See \cite{Trotter} for details.} seems necessary as suggested by hardness results of classical simulation of quantum circuits (e.g., \cite{TerDiv:J04}).
It was also previously known that a classical algorithm can estimate local observables in time exponential in the size of the light-cone~\cite{Hastings04,Osborne06}, i.e., the number of input qubits that could influence a particular output qubit, resulting in a similar run-time bound. (For Hamiltonian evolution we still have an effective light-cone due to Lieb-Robinson bounds, e.g., \cite{haah2018quantum}.)
Our approach is, however, strictly stronger in the sense that we could estimate correlations across the entire system, something that cannot be achieved by the light-cone argument.

\begin{figure}
\begin{tikzpicture}[thick, > = latex,
  semi/.style = {semithick},
  pt/.style = {circle, draw = black, fill = black, inner sep = 1pt},
  nd/.style = {circle, draw = black, fill = black, inner sep = 0pt}
]

\def\W{4pt}
\def\O{56pt}

\def\d{4pt}
\foreach \t in {0,1,2,3}{
  \fill[blue!20, rounded corners = 2pt] (\O*\t-\d,-\d) rectangle +(8*\W+2*\d,8*\W+2*\d);
}

\draw[red] (8*\W,8*\W) -- (1*\O,8*\W);
\draw[red] (8*\W+1*\O,8*\W) -- (2*\O,8*\W);
\draw[red] (8*\W+2*\O,8*\W) -- (3*\O,8*\W);

\draw[red] (8*\W,6*\W) -- (1*\O, 6*\W);
\draw[red] (8*\W+1*\O,6*\W) -- (2*\O,6*\W);
\draw[red] (8*\W+2*\O,6*\W) -- (3*\O,6*\W);

\foreach \t in {0,1,2,3}{

	\foreach \x in {0,2,6,8}{
  		\foreach \y in {0,2,6,8} {
    		\node[pt] at (\x*\W+\t*\O,\y*\W) {};
  		}
  		\foreach \y in {3,4,5} {
  			\node[nd] at (\x*\W+\t*\O,\y*\W) {};
  		}
  		\foreach \y in {0,6}{
  			\draw (\x*\W+\t*\O,\y*\W) -- (\x*\W+\t*\O,\y*\W+2*\W);
  		}
  }

  \foreach \y in {0,2,6,8}{
  		\foreach \x in {3,4,5}{
  				\node[nd] at (\x*\W+\t*\O,\y*\W) {};
  		}
  		\foreach \x in {0,6}{
  				\draw (\x*\W+\t*\O,\y*\W) -- (\x*\W+2*\W+\t*\O,\y*\W);
  		}
  }

}

\def\L{7pt};


\node at (6*\W,-2.8*\W) {$\sqrt{n} \times \sqrt{n}$ grid};

\def\C{8*\W+\O+2*\L};

\end{tikzpicture}
\caption{Interaction graph $G$ of a local Hamiltonian
with four parties, each a square grid of size $\sqrt{n} \times \sqrt{n}$. Each pair of adjacent parties has a weak interaction, indicated by the red lines.  Since the contraction complexity $\cc(g)$ of the induced graph $g$ and the interaction strength $h$ are both $O(1)$, for short periods of time (e.g., $t=O(1)$), this $4n$-qubit system can be efficiently simulated on an $n$-qubit quantum computer. }
\label{fig:Hamiltonian}
\end{figure}
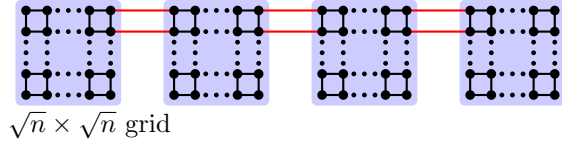

\newcommand{\ttheta}{\theta}
\sect{\Lone{Application to VQE}}%
Variational Quantum Eigensolver (VQE) is a variational method for finding the lowest eigenvalue of an $n$-qubit Hamiltonian $H$ by applying some parameterized circuit $U(\ttheta)$ to $\ket{0}\xp{n}$ and minimizing the expectation with $H$: $ \min_{\ttheta} \; \bra{0}\xp{n} U(\ttheta)^{\dag} H U(\ttheta) \ket{0}\xp{n}.$
\Cb{
This method has been proposed for solving optimization problems on quantum computers \cite{wecker2015progress, mcclean2016theory, farhi2014quantum} and, thanks to its short-depth circuits, has become a promising candidate to surpass the classical optimization methods and show quantum advantage on NISQ devices \cite{preskill2018quantum, peruzzo2014variational, o2016scalable, kandala2017hardware, moll2017quantum}.
}
\newcommand{\Uent}{U_\text{ENT}}

\Cb{
In \cite{kandala2017hardware}, Kandala et al.\ propose a class of \textit{hardware-friendly} variational circuits $U(\ttheta)$ and experimentally demonstrate the effectiveness of their VQE method for addressing problems of small molecules and quantum magnetism, using up to $6$ qubits. Their ansatz $U(\ttheta)$ has the following form:
\begin{equation}
  U(\ttheta) = U_D(\theta_D) \Uent \dotsb U_1(\theta_1) \Uent U_0(\theta_0),
\end{equation}
where $U_i(\theta_i) = \bigotimes_{j=1}^n U_i^j(\theta_i^j)$ and each $U_i^j(\theta_i^j)$ is a parameterized single-qubit gate applied on the $j$-th qubit out of $n$, $\Uent$ is a fixed sequence of two-qubit gates meant for producing entanglement, and $D$ is the number of rounds. }

\Cb{
In the context of current NISQ devices, we propose a way to reduce the number of qubits required for implementing $U(\ttheta)$ by using our cluster simulation scheme.
This involves the following steps:
(i)~choosing a partition $\Part = \set{S_1, \dotsc, S_r}$ of $n$ qubits such that $\abs{S_i} \leq d$ for each $i$;
(ii)~removing some entangling gates from $\Uent$ that go across different parts of $\Part$ to decrease $\cc(g)$, where $g$ is the graph induced by regarding each set $S_i$ as a node and each gate that acts across two sets as an edge;
(iii) runing this $n$-qubit $U(\theta)$ using a $(2^{O(\cc(g))},d)$-simulator.
}

\Cb{
We report an experiment estimating the ground energy of the $\text{BeH}_2$ molecule; see \cref{fig:VQE}. Using a 3-qubit physical device, we run the 6-qubit $U(\ttheta)$ from~\cite{kandala2017hardware} and achieve the similar accuracy, thus demonstrating the potential of implementing VQE with limited quantum memory. Additional details about the experiment and the discussion of reducing $\cc(g)$ can be found in the Supplemental Material.
}

\begin{figure}
\includegraphics[scale=0.45]{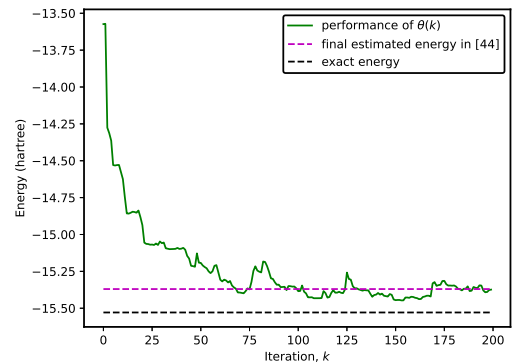}
\caption{
\Cb{
Estimating the ground energy of $\text{BeH}_2$ with interatomic distance of 1.7~\AA{} by running the 6-qubit VQE of \cite{kandala2017hardware} on “ibmq ourense”, a 5-qubit device provided by the IBM Quantum Experience \cite{IBM:20}. We use up to three qubits of the device. At the $k$th step, we employ the iterative optimization to update $\theta(k)$ based on the simultaneous perturbation stochastic approximation (SPSA) method, similar as \cite{kandala2017hardware}. In particular, each run of the 6-qubit ansatz $U(\ttheta_{k})$ with $D=1$ layers is simulated by executing 12 different 3-qubit circuits. See the source code at \url{https://github.com/TianyiPeng/Partiton_VQE}. 
}
}
\label{fig:VQE}
\end{figure}

\sect{\Lone{Summary}}%
In this Letter, we provide a systematic approach for simulating clustered quantum circuits with limited use of quantum memory. \CR{Our scheme is relevant to promising NISQ applications such as Hamiltonian simulation and VQE. By reducing the number of qubits and the depth of the circuit, it is particularly applicable to intermediate scale devices and potentially also improves the circuit's robustness to correlated noise.} We leave open the problem of determining the best $(K,d)$ or $(\cc(g),d)$ for a given quantum circuit (this may be related to the graph partitioning, graph clustering, and treewidth problems). \CR{Another direction is to develop more case-by-case optimization techniques for realistic applications under our scheme.}


\begin{acknowledgements}
We thank Robin Kothari, Shuhua Li, Xiao Yuan, and Yuan Su for helpful discussions. We thank Linsen Li, Kaidong Peng, Yufeng Ye, Zhen Guo for the help on experiments. Part of this work was done while MO and XW were visiting MIT.
XW is supported by NSF CCF-1755800, CCF-1816695 and CCF-1942837. MO acknowledges Leverhulme Trust Early Career Fellowship (ECF-2015-256) and NWO Vidi grant VI.Vidi.192.109 for financial support.  TP acknowledges support from the Top Open program in Tsinghua University, China.
AWH was funded by NSF grants CCF-1452616, CCF-1729369, and  PHY-1818914; ARO contract W911NF-17-1-0433; and
the MIT-IBM Watson AI Lab under the project {\it Machine Learning in Hilbert space}.
\end{acknowledgements}


\bibliographystyle{apsrev4-1}
\bibliography{references}

%
%
%
%
%
%
%
%
%
%
%



\section{Supplemental Material}

\subsection{I: Table of Notation}

\begin{table}[h]
\centering
\begin{tabular}{|c|c|}
\hline
\textbf{Notation} & \textbf{Definition}\\
\hline
$C$ & the quantum circuit\\
\hline
$f$ & the classical post-processing function\\
\hline
\multirow{2}{*}{$y$}
  & the $n$-bit string obtained\\
  & when measuring all qubits\\
\hline
\multirow{2}{*}{$G$}
  & the graph associated with\\
  & the circuit (or the Hamiltonian)\\
\hline
\multirow{2}{*}{$\A$}
  & the set of tensors associated\\
	& with the graph $G$\\
\hline
\multirow{2}{*}{$K$}
  & the number of edges\\
  & between different clusters of $G$\\
\hline
\multirow{2}{*}{$d$}
  & the number of qubits\\
	& sufficient for simulating each cluster\\
\hline
\multirow{2}{*}{$g$}
  & the graph obtained by contracting\\
  & each cluster of $G$ to a single vertex\\
\hline
$\cc(g)$ & the contraction complexity of $g$\\
\hline
\end{tabular}
\end{table}

\subsection{II: Tensor networks}

Any quantum circuit $C$ can be represented by a directed graph $G = (V,E)$ with three types of vertices: input states (denoted by~$\lhd$), quantum gates (denoted by~$\Box$), and observables (denoted by~$\rhd$),
with directed edges indicating the flow of qubits between them.
We can further associate to each vertex $v \in V$ a tensor $A(v)$ that encodes the matrix entries of the object represented by this vertex,
and denote by $\A = \set{A(v) : v \in V}$ the collection of all these tensors.
This results in a tensor network $(G,\A)$ that completely captures the algorithm described by the original quantum circuit $C$~\cite{Vidal03,SDV06,Jozsa,MS08}.
In particular, any QC algorithm $(C,f)$ can also be represented by such a tensor network when we combine the final computational basis measurement and the classical post-processing function $f$ into a diagonal observable $O_f$ (see Figure~1, Phase~2 in the main paper).

Our index convention for $k$-qubit tensors is borrowed from~\cite{MS08}. Consider $\alpha = (\alpha_1, \dotsc, \alpha_k)$ where each $\alpha_j = (\alpha_j^1, \alpha_j^2)$ is a pair of binary indices, i.e., $\alpha_j \in \Pi = \set{0,1}^2$ \footnote{We need two indices per qubit since we are working with an operator rather than a vector basis.}. Let $M(\alpha_j) = \ket{\alpha_j^1} \bra{\alpha_j^2}$ be the elementary matrix corresponding to $\alpha_j$ and let $M(\alpha) = \bigotimes_{j=1}^k M(\alpha_j) = \ket{\alpha_1^1, \dotsc, \alpha_k^1} \bra{\alpha_1^2, \dotsc, \alpha_k^2}$ be its $k$-qubit generalization. For a state $\rho$, gate $U$, and observable $O$, the entries of their tensors are computed as follows:
\begin{align}
  A(\rho)_\alpha
  &= \Tr \sof[\big]{
       \rho \cdot M(\alpha)\ct
     }, \label{eq:AR} \\
  {}_\alpha A(U)_\beta
  &= \Tr \sof[\big]{
       U M(\alpha)
       U^{\dag} \cdot
       M(\beta)\ct
     }, \label{eq:AU} \\
  {}_\beta A(O)
  &= \Tr \sof[\big]{
       M(\beta) \cdot O
     }, \label{eq:AO}
\end{align}
where the location of subscripts indicates whether qubits are incoming or outgoing. For example,
\begin{equation}
  \Tr \sof[\big]{U \rho U\ct O}
  = \sum_{\alpha,\beta} A(\rho)_\alpha \cdot {}_\alpha A(U)_\beta \cdot {}_\beta A(O).
\end{equation}
This corresponds to the following tensor network:
\begin{center}
\begin{tikzpicture}[thick, > = latex,
    gate/.style = {fill = white, draw, text height = 1.5ex, text depth = .25ex, minimum size = 18pt},
    tr/.style = {isosceles triangle, inner sep = 0pt, minimum width = 22pt, isosceles triangle apex angle = 70},
    ob/.style = {tr, draw},
    st/.style = {tr, draw, shape border rotate = 180},
  ]
  \node[st]   (R) at (0,0) {$\rho$};
  \node[gate] (U) at (2,0) {$U$};
  \node[ob]   (O) at (4,0) {$O$};
  \draw[->] (R) -> (U);
  \draw[->] (U) -> (O);
  \node at (1,0.2) {$\alpha$};
  \node at (3,0.2) {$\beta$};
\end{tikzpicture}
\end{center}
where arrows indicate the flow of information.

For a general tensor network $(G,\A)$, its \emph{value} is
\begin{equation}
  T(G,\A) = \sum_{\alpha \in \Pi^E} \prod_{v\in V} A(v)_{\alpha(v)},
  \label{eq:TGA}
\end{equation}
where $\alpha = (\alpha_e : e \in E)$ and each $\alpha_e = (\alpha_e^1, \alpha_e^2) \in \Pi$ ranges over all possible index assignments~\footnote{Equivalently, one can think of each $\alpha \in \Pi^E$ as a function of the form $\alpha: E \to \Pi$.} for edge $e$, and $\alpha(v)$ denotes the restriction of $\alpha$ to edges incident to $v$. In particular, $A(v)_{\alpha(v)}$ is the entry of $A(v)$ when all edges incident to $v$ are labeled according to $\alpha$~\footnote{Our notation in \cref{eq:TGA} ignores the distinction between incoming and outgoing edges, but one has to take this into account when evaluating the tensor entries using \cref{eq:AR,eq:AU,eq:AO}.}. Note that
\begin{equation}
  T(G,\A) = \E_y f(y),
  \label{eq:TGAEf}
\end{equation}
so a quantum circuit can be simulated by approximating the value of its tensor network~\cite{MS08}.

\subsection{III: Proof of \cref{lem:edgecut} \label{proof:lem:edgecut}}

\setcounter{lemma}{0}

\begin{lemma}\label{lem:edgecut}
Let $(G(E, V),\A)$ be a tensor network of a QC algorithm and let $uv$ be any edge in $G$. There is a set of eight coefficients $c_i \in \set{-\frac{1}{2},\frac{1}{2}}$, observables $O_i$, and states $\rho_i$, such that the following modification of the edge $uv$
\begin{center}
  \begin{tikzpicture}[thick, > = latex,
      gate/.style = {fill = white, draw, text height = 1.5ex, text depth = .25ex, minimum size = 18pt},
      pt/.style = {circle, draw = black, fill = black, inner sep = 1.5pt},
      tr/.style = {isosceles triangle, inner sep = 0pt, minimum width = 28pt, isosceles triangle apex angle = 70},
      ob/.style = {tr, draw},
      st/.style = {tr, draw, shape border rotate = 180}
    ]
    \node[pt] (u) at (-4.0,0) {}; \node at (-4.0,0.3) {$u$};
    \node[pt] (v) at (-2.5,0) {}; \node at (-2.5,0.3) {$v$};
    \draw[->] (u) -- (v);
    \node at (-1.7,0) {$\Longrightarrow$};
    \node at (-0.7,0) {$\displaystyle\sum_{i=1}^8 c_i$};
    \node[pt] (U) at (0.0,0) {}; \node at (0,0.3) {$u$};
    \node[ob] (O) at (0.9,0) {$O_i$};
    \node[st] (R) at (2.1,0) {$\rho_i$};
    \node[pt] (V) at (3.0,0) {}; \node at (3,0.3) {$v$};
    \draw[->] (U) -- (O);
    \draw[->] (R) -- (V);
  \end{tikzpicture}
\end{center}
does not affect the value of the overall tensor network. Moreover, each term above yields some tensor network $(G',\A_i)$ that again corresponds to a valid quantum circuit, and their values are related as follows:
\begin{equation}
  T(G,\A) = \sum_{i=1}^8 c_i T(G',\A_i).
\end{equation}
\end{lemma}

\begin{proof}
Let $A$ be any $2 \times 2$ matrix and let
\begin{align*}
  I &= \smx{1&0\\0&1},  &
  X &= \smx{0&1\\1&0},  &
  Y &= \smx{0&-i\\i&0}, &
  Z &= \smx{1&0\\0&-1}
\end{align*}
be the Pauli matrices. Since the normalized Pauli matrices $\set{I,X,Y,Z}/\sqrt{2}$ form an orthonormal matrix basis,
\begin{equation*}
  A
  = \frac{
    \Tr(A I) I
  + \Tr(A X) X
  + \Tr(A Y) Y
  + \Tr(A Z) Z}{2}.
\end{equation*}

Let us expand each Pauli matrix in its eigenbasis. If we denote the Pauli marices by $O_i$, their eigenprojectors by $\rho_i$ and the corresponding eigenvalues by $2c_i$ so that $O_1 = O_2 = 2 (c_1 \rho_1 + c_2 \rho_2)$, etc., then
\begin{align}
  O_1 &= I, & \rho_1 &= \proj{0}, & c_1 &= +1/2, \label{eq:1} \\
  O_2 &= I, & \rho_2 &= \proj{1}, & c_2 &= +1/2, \\
  O_3 &= X, & \rho_3 &= \proj{+}, & c_3 &= +1/2, \\
  O_4 &= X, & \rho_4 &= \proj{-}, & c_4 &= -1/2, \\
  O_5 &= Y, & \rho_5 &= \proj{{+}i}, & c_5 &= +1/2, \\
  O_6 &= Y, & \rho_6 &= \proj{{-}i}, & c_6 &= -1/2, \\
  O_7 &= Z, & \rho_7 &= \proj{0}, & c_7 &= +1/2, \\
  O_8 &= Z, & \rho_8 &= \proj{1}, & c_8 &= -1/2, \label{eq:8}
\end{align}
where $\ket{\pm} = \frac{\ket{0} \pm \ket{1}}{\sqrt{2}}$ and $\ket{{\pm}i} = \frac{\ket{0} \pm i \ket{1}}{\sqrt{2}}$.

For each case, define a corresponding map
\begin{equation}
  \Phi_i(A) = \Tr(A O_i) \rho_i
\end{equation}
that can be implemented by first measuring the Pauli observable $O_i$ and then preparing the corresponding eigenstate $\rho_i$. It follows immediately that
\begin{equation}
  A = \sum_{i=1}^8 c_i \Phi_i(A).
  \label{eq:sum}
\end{equation}

If $\Phi$ denotes the trivial single-qubit channel that acts as $\Phi(A) = A$, for any $2 \times 2$ matrix $A$, then \cref{eq:sum} is equivalent to
\begin{equation}
  \Phi = \sum_{i=1}^8 c_i \Phi_i.
  \label{eq:Phi}
\end{equation}
The same decomposition works even when $\Phi$ acts on a subsystem of a larger system.

To prove \cref{lem:edgecut}, note that adding the trivial map $\Phi$ on any edge of $G$ does not affect the value of the tensor network $(G,\A)$. Thus our strategy for constructing the decomposition $\set{(G',\A_i)}_{i=1}^8$ is to explicitly include the trivial map $\Phi$ on the edge $uv$, and then decompose $\Phi$ according to \cref{eq:Phi}. Intuitively, we are cutting the edge $uv$ and mending its both ends with additional tensors in several different ways. This gives us a collection of eight new tensor networks $(G',\A_i)$ that each correspond again to a quantum circuit: the \mbox{$i$-th} circuit measures the qubit produced at vertex $u$ with observable $O_i$ and supplies the qubit required at vertex $v$ in state $\rho_i$.

To guarantee that the modified tensor network has the same value as the original, we have to verify that
\begin{center}
\begin{tikzpicture}[thick, > = latex,
    gate/.style = {fill = white, draw, text height = 1.5ex, text depth = .25ex, minimum size = 18pt},
    pt/.style = {circle, draw = black, fill = black, inner sep = 1.5pt},
    tr/.style = {isosceles triangle, inner sep = 0pt, minimum width = 28pt, isosceles triangle apex angle = 70},
    ob/.style = {tr, draw},
    st/.style = {tr, draw, shape border rotate = 180}
  ]
  \node (u) at (-4.3,0) {}; \node at (-4.0,0.3) {$\alpha$};
  \node (v) at (-2.2,0) {}; \node at (-2.5,0.3) {$\beta$};
  \node (mm) at (-3.44, 0) {};
  \draw[->] (u) -- (mm);
  \draw[->] (mm) -- (v);
  \node[gate] (I) at (-3.25,0) {$I$};
  \node at (-1.7,0) {$=$};
  \node at (-0.7,0) {$\displaystyle\sum_{i=1}^8 c_i$};
  \path (U) coordinate (0,0); \node at (0.3,0.3) {$\alpha$};
  \node[ob] (O) at (0.9,0) {$O_i$};
  \node[st] (R) at (2.1,0) {$\rho_i$};
  \path (V) coordinate (3,0); \node at (2.7,0.3) {$\beta$};
  \draw[->] (U) -- (O);
  \draw[->] (R) -- (V);
\end{tikzpicture}
\end{center}
where on the left-hand side we use conjugation by the identity matrix $I$ as a way of implementing the trivial channel $\Phi$. Following the notation of \cref{eq:AR,eq:AU,eq:AO}, this is equivalent to verifying
\begin{equation}
  {}_\alpha A(I)_\beta = \sum_{i=1}^8 c_i \cdot {}_\alpha A(O_i) \cdot A(\rho_i)_\beta.
  \label{eq:Iab}
\end{equation}
According to \cref{eq:AO}, the left-hand side is
\begin{equation}
  {}_\alpha A(I)_\beta
  = \Tr \sof[\big]{ M(\alpha) M(\beta)\ct } = \delta_{\alpha,\beta}.
\end{equation}
From \cref{eq:AR,eq:AO}, the right-hand side is
\begin{equation}
  \sum_{i=1}^8 c_i
  \cdot \Tr \sof[\big]{ M(\alpha) O_i }
  \cdot \Tr \sof[\big]{ \rho_i M(\beta)\ct }.
\end{equation}
We can also write this as
\begin{equation}
  \Tr \sof[\big]{ \of[\big]{M(\alpha) \x M(\beta)\ct} \cdot S }
\end{equation}
where $S = \sum_{i=1}^8 c_i (O_i \x \rho_i)$. Note from \cref{eq:1} to \cref{eq:8} that
\begin{align}
  S &= \frac{1}{2} \of{ I \x I + X \x X + Y \x Y + Z \x Z} \\
    &= \mx{
         1 & 0 & 0 & 0 \\
         0 & 0 & 1 & 0 \\
         0 & 1 & 0 & 0 \\
         0 & 0 & 0 & 1
       }
\end{align}
is the two-qubit SWAP operator, so
\begin{align}
  \Tr \sof[\big]{ \of[\big]{M(\alpha) \x M(\beta)\ct} \cdot S }
  &= \Tr \sof[\big]{ M(\alpha) M(\beta)\ct } \\
  &= \delta_{\alpha,\beta}
\end{align}
which proves \cref{eq:Iab}.
\end{proof}

\subsection{Hoeffding's inequality}

Our proofs of \cref{thm:main,thm:Treewidth} rely on the convergence speed of averaged i.i.d.\ random variables. They use the following inequality from \cite{Hoeffding}, which we state here for the sake of reference.
\begin{claim}
Suppose $X_1,\dotsc,X_N$ are i.i.d.\ real random variables such that $|X_i| \leq a$ and let $\bar{X} = (X_1 + \dotsb + X_N) / N$. Then
\begin{equation*}
  \Pr \sof*{\abs*{\bar{X} - \E[\bar{X}]} < \delta}
  \geq 1 - 2 \exp \of*{-\frac{N\delta^2}{2a^2}}
  > \frac{2}{3},
\end{equation*}
where the last inequality holds if $N \geq 4 a^2 / \delta^2$.
\end{claim}

\subsection{IV: Proof of \cref{thm:main}}

\setcounter{theorem}{0}

\begin{theorem}\label{thm:main}
Any QC algorithm $(C,f)$ with a $(K,d)$-clustered circuit $C$ consisting of $r$ clusters can be simulated with precision $\epsilon$ by running a $d$-qubit quantum device $O\left(2^{4K} r/\epsilon^2\right)$ times.
The total classical and quantum run-time of this simulation is $O\left(2^{4K}(n+m)/\epsilon^2\right)$ where $n$ and $m$ are the number of qubits and gates in $C$, respectively.
\end{theorem}

\begin{proof}
From the discussion leading up to \cref{eq:TGAEf} we know that simulating a quantum circuit is equivalent to approximating $T(G,\A)$, the value of its tensor network given in \cref{eq:TGA}. Let $\Ecut$ denote the set of edges that connect different clusters of $G$. Since $G$ is $(K,d)$-clustered, $\abs{\Ecut} = K$. Let us decompose each edge $e \in \Ecut$ according to \cref{lem:edgecut} and denote the modified graph by $G'$. If we collect all indices appearing in this decomposition into a single string $s = (s_e : e \in \Ecut)$, where $s_e \in [8] = \set{1,\dotsc,8}$ represents the index $i$ of edge $e \in \Ecut$, then
\begin{align}
  T(G,\A) &= \sum_{s \in [8]^{\Ecut}} c_s T(G',\A_s), \label{eq:decomp} \\
  c_s &= \prod_{e \in \Ecut} c^e_{s_e}, \label{eq:cs}
\end{align}
where $c^e_1, \dotsc, c^e_8 \in \set{-\frac{1}{2}, \frac{1}{2}}$ are the coefficients $c_i$ from \cref{lem:edgecut}. Regarding $s$ as a random variable sampled uniformly from $[8]^{\Ecut}$,
\begin{equation}
  T(G,\A) = \E_s \sof[\big]{8^K c_s T(G',\A_s)}.
  \label{eq:Es}
\end{equation}

To estimate the value of $T(G',\A_s)$, for a given $s \in [8]^{\Ecut}$, recall from \cref{lem:edgecut} that each $(G',\A_s)$ is also a quantum circuit. Moreover, since $G$ is $(K,d)$-clustered, each $(G',\A_s)$ consists of $r$ pieces (not counting the final observable $O_f$), each of which can be executed on a $d$-qubit quantum computer. Our simulator runs these pieces sequentially. To obtain the correct result, we need to initialize and measure all qubits of each piece in the right basis. Our strategy is as follows. We initialize an input qubit of a given piece to $\ket{0}$ if it corresponds to an input qubit also in the original circuit. Similarly, we measure an output qubit in the standard basis if it is originally fed into the final observable $O_f$. The remaining input and output qubits are supposed to be communicated between different clusters. In the simulation, however, we respectively initialize them to Pauli eigenstates and measure using Pauli observables determined by $s$. Specifically, if $e \in \Ecut$ is an edge between two clusters, its endpoints correspond to a pair of input / output qubits that are initialized / measured according to the value of $s_e \in [8]$ (see the proof of \cref{lem:edgecut}).

After executing all pieces as described above, we obtain a collection of Pauli measurement outcomes $\sigma^e_{s_e} \in \set{-1,1}$, one for each $e \in \Ecut$, and a collection of strings $y_1, \dots, y_r$ of total length $n$~\footnote{Some of these strings might be empty. This happens when all outgoing qubits from some cluster are fed into other clusters and none are fed directly into the final measurement.}. We can now compute the following quantity by evaluating the classical post-processing function $f$:
\begin{equation}
  t_s = f(y_1, \dots, y_r)\prod_{e\in \Ecut} \sigma^e_{s_e}.
\end{equation}
Note that, for each $s \in [8]^{\Ecut}$, $t_s \in [-1,1]$ is a random variable with expectation $T(G',\A_s)$. Hence, using \cref{eq:Es} we can express the value of $T(G,\A)$ as an expectation over a combined random variable $(s,t_s)$, where $s$ is drawn uniformly from $[8]^{\Ecut}$ and $t_s$ is obtained by running the corresponding quantum circuit as described above:
\begin{equation}
  T(G,\A) = \E_{(s,t_s)} \sof{8^K c_s t_s}.
\end{equation}

Our simulation approximates $T(G,\A)$ by producing $N$ samples $\set{(s_1, t_{s_1}), \dotsc, (s_N, t_{s_N})}$ of the random variable $(s, t_s)$ and computing their average $\frac{1}{N} \sum_{i=1}^N 8^{K} c_{s_i} t_{s_i}$. The magnitude of each term is bounded by
\begin{equation}
  |8^{K}c_{s_i}t_{s_i}| \leq 8^{K} \frac{1}{2^{K}} = 2^{2K} = a,
\end{equation}
where we used $|c_{s}| = \prod_{e \in \Ecut} |c_{s_e}^{e}| = 1/2^K$. By Hoeffding's inequality (see above), we can achieve accuracy $\epsilon$ with success probability at least $\frac{2}{3}$ if the number of experiments $N$ satisfies
\begin{equation}
  N = 4a^2/\epsilon^2 = O(2^{4K}/\epsilon^2).
\end{equation}
Since each sample of $t_s$ requires running $r$ independent quantum circuits, where $r$ is the number of clusters, the total number of circuits run on the $d$-qubit quantum device is
\begin{equation}
  N \cdot r = O(2^{4K} r/\epsilon^2).
\end{equation}

To bound the total running time, let us denote the number of qubits and gates in the $i$-th cluster by $n_i$ and $m_i$, respectively.
The total processing time for simulating the $i$-th cluster once is $O(n_i + m_i)$.
Note that $\sum_{i=1}^{r} n_i + m_i = O(n+m)$, where $n$ is the total number of qubits and $m$ is the total number of gates of $C$.
Hence the total complexity is
\begin{equation}
  O(2^{4K} (n+m) /\epsilon^2),
\end{equation}
as claimed.
\end{proof}

\subsection{V: Proof of \cref{thm:Treewidth}}

\begin{theorem}\label{thm:Treewidth}
An $n$-qubit and m-gate QC algorithm $(C,f)$ with a decomposable function $f$ and a $(K,d)$-clustered circuit $C$ can be simulated with precision $\epsilon$ by running a $d$-qubit quantum device $O\of[\big]{2^{O(d'(g))} (r^3 \log r) / \epsilon^2}$ times.
Here $g$ denotes the induced (multi-)graph obtained from $G$ by removing the vertex $O_f$ and contracting each cluster to a single vertex while keeping the edges between different clusters, $d'(g)$ is the maximum degree of a vertex in $g$, and $r$ is the number of vertices in $g$ or clusters in $G$ (note that $r\leq n+m$).
The total running time of this simulation is $T_Q + T_C$ where
\begin{align}
  T_Q &= O\of[\big]{2^{O(d'(g))} (r^2 \log r)(n+m)/\epsilon^2}, \\
  T_C &= 2^{O(\cc(g))} \poly(r)
\end{align}
are the total quantum and classical running times, and $\cc(g)$ is the contraction complexity of $g$.
\end{theorem}

\begin{proof}
Since $f$ is decomposable, $f(y) = \prod_{j=1}^r f_j(y_j)$, so we can replace the final observable $O_f$ by $r$ independent smaller observables $O_{f_1}, \dotsc, O_{f_r}$. Moreover, since each $O_{f_j}$ is connected to a different cluster, we can absorb it into the corresponding cluster. Let us denote this slightly modified tensor network by $(G,\A)$ and note that it is also $(K,d)$-clustered just like the original one. If we cut each edge in $\Ecut$ using \cref{lem:edgecut} (just like we did in the proof of \cref{thm:main}), the resulting decomposition is again given by \cref{eq:decomp,eq:cs}. However, this time the graph $G'$ of the modified network fully breaks up into individual clusters \footnote{Recall that $\Ecut$ does \emph{not} include the edges attached to $O_f$, so the original tensor network did not fully break up because the clusters connected to $O_f$ remained in the same connected component even when the edges in $\Ecut$ were removed.}:
\begin{center}
\begin{tikzpicture}[semithick,
    pt/.style = {circle, draw = black, fill = black, inner sep = 1.5pt},
    PT/.style = {circle, draw = black, fill = black, inner sep = 2.5pt},
  	cr/.style = {circle, draw = black, fill = white, inner sep = 1.5pt}
	]

  \begin{scope}[xshift = -1cm]
    \draw (0,0) circle [radius = 0.5];
    \node at (-0.75,-0.1) {$S_2$};
    \node[pt] (a1) at ( 0.2,-0.2) {};
    \node[pt] (a2) at (-0.2,-0.2) {};
    \node[pt] (a3) at (-0.2, 0.2) {};
    \node[pt] (a4) at ( 0.2, 0.2) {};
    \draw (a1) -- (a2) -- (a3) -- (a4) -- (a1);
    \draw (a1) -- (a3);
    \draw (a2) -- (a4);
  \end{scope}

  \begin{scope}[xshift = 1cm]
    \draw (0,0) circle [radius = 0.5];
    \node at (0.75,-0.1) {$S_3$};
    \node[pt] (b1) at (-0.2,-0.18) {};
    \node[pt] (b2) at ( 0.2,-0.18) {};
    \node[pt] (b3) at ( 0.0, 0.22) {};
    \draw (b1) -- (b2) -- (b3) -- (b1);
  \end{scope}

  \begin{scope}[yshift = 1.7cm]
    \draw (0,0) circle [radius = 0.5];
    \node at (0.75,0.1) {$S_1$};
    \node[pt] (c1) at ( 0.2, 0.2) {};
    \node[pt] (c2) at (-0.2, 0.2) {};
    \node[pt] (c3) at (-0.2,-0.2) {};
    \node[pt] (c4) at ( 0.2,-0.2) {};
    \draw (c1) -- (c2) -- (c3) -- (c4) -- (c1);
    \draw (c1) -- (c3);
  \end{scope}

  \foreach \s/\t in {a1/b1, b3/c4, c3/a4} {
    \draw (\s) -- (\t);
    \node[pt,white] at ($(\s)!0.5!(\t)$) {};
    \node[cr] at ($(\s)!0.4!(\t)$) {};
    \node[cr] at ($(\s)!0.6!(\t)$) {};
  }

  \node at (-3.5, 0.7) {$T(G,\A) \; = \; \displaystyle\sum_{s \in [8]^{\Ecut}} c_s$};
  \node at ( 0.0,-0.8) {$T(G',\A_s)$};

\end{tikzpicture}
\end{center}
If we denote these clusters by $S_1, \dotsc, S_r$, we can further decompose each term as a product over them:
\begin{equation}
  T(G',\A_s)
  = \prod_{j=1}^r T \of[\big]{ G'^j, \A^j_{s(S_j)} },
  \label{eq:prod}
\end{equation}
where $G'^j$ is the $j$-th cluster (or the $j$-th component) of $G'$ and $s(S_j)$ denotes the restriction of $s$ to the edges incident to cluster $S_j$ (just like in \cref{eq:TGA}). In particular, note that $\bigcup_{j=1}^r \A^j_{s(S_j)} = \A_s$ and that each $\of[\big]{ G'^j, \A^j_{s(S_j)} }$ is a piece of the larger network $(G', \A_s)$ as well as a tensor network on its own.

Operationally, the value $T\of[\big]{G'^j,\A^j_{s(S_j)}}$ in \cref{eq:prod} is the expectation of the circuit of the $j$-th component of $(G',\A_s)$ whose all input qubits are initialized to specific eigenvectors of Paulis (determined by $s$) and all output qubits are measured by specific Pauli observables (also determined by $s$) or the observables $O_{f_j}$ whose eigenvalues lay in $[-1,1]$. As a consequence,
\begin{equation}
  \abs[\big]{T\of[\big]{G'^j,\A^j_{s(S_j)}}} \leq 1.
  \label{eq:bounded}
\end{equation}

We define a new (multi)-graph $g$ obtained by contracting each $S_j$ of graph $G$ to a single node. The vertices and edges of the graph $g$ can be identified with $\set{S_1, \dotsc, S_r}$ and $\Ecut$, respectively. For simplicity, let us abbreviate the vertex $S_j$ as $j$. Then, we define a new tensor network $(g, \a)$ associated with the graph $g$. For each vertex $j$ of $g$, we assign a tensor $a(j) \in \a$ whose indices range over $[8] = \set{1,\dotsc,8}$ and whose entries are given by
\begin{equation}
  a(j)_{s(j)}
  = c_{s(j)} \;
    T \of[\big]{ G'^j, \A^j_{s(j)} },
  \label{eq:asj}
\end{equation}
where $s(j)$ is the restriction of $s \in [8]^{\Ecut}$ to $S_j$ and $c_{s(j)}$ is the product of $c^e_{s_e}$ (see \cref{eq:cs}), with $e$ ranging over all edges in $\Ecut$ connecting $S_j$ and any $S_{j'}$ with $j'>j$. Note that the indices of edges in $(g, \a)$ range over $[8]=\set{1,\dotsc, 8}$, while the indices of edges in $(G,\A)$ range over $\Pi = \set{0,1}^2$.

Intuitively, $a(j)_{s(j)}$ is obtained by contracting the $j$-th cluster of $\A_s$ and absorbing the coefficients $c^e_{s_e}$ into the cluster with the smaller index (each edge $e \in \Ecut$ connects two different clusters and we always absorb into the one with the smallest index). Following \cref{eq:TGA}, the value of $(g,\a)$ is
\begin{align}
  T(g,\a)
  &= \sum_{s \in [8]^{\Ecut}} \prod_{j=1}^r a(j)_{s(j)} \\
  &= \sum_{s \in [8]^{\Ecut}} c_s \prod_{j=1}^r T(G'^j, \A^j_{s(j)}) \\
  &= \sum_{s \in [8]^{\Ecut}} c_s T(G', \A_s) \\
  &= T(G,\A),
\end{align}
where we used equations \eqref{eq:asj}, \eqref{eq:prod}, and \eqref{eq:decomp}, and the fact that each $c^e_{s_e}$ occurs exactly once (from the cluster with the smallest index out of the two connected by the edge $e$). Because of this identity, we can instead focus on computing the value of the smaller tensor network $(g,\a)$.

Our method of computing $T(g,\a)$ starts by estimating all entries $a(j)_{s(j)}$ of tensors in $\a = \set{a(1), \dotsc, a(r)}$, which can be done by running a $d$-qubit quantum computer multiple times to estimate $T(G'^j, \A^j_{s(j)})$. This yields an approximation $\tilde{\a}$ of $\a$. Then we contract $(g,\tilde{\a})$ using a classical algorithm. To obtain $T(g,\a)$ with precision $\epsilon$, we need to compute the entries of each $a(j)$ to precision roughly $2^{-O(K)}$ where $K$ is the number of edges of $g$. We can improve this trivial estimate by using the fact that these tensors come from quantum circuits.

Recall from \cref{eq:bounded} that $\abs[\big]{T(G'^j, \A^j_{s(j)})} \leq 1$, for any $j \in \set{1, \dotsc, r}$ and $s \in [8]^{\Ecut}$. Hence,
\begin{equation}
  \abs{a(j)_{s(j)}} \leq 1
  \label{eq:abound}
\end{equation}
from the definition in \cref{eq:asj}. More generally, consider a new tensor $a(L)$ obtained by contracting any subset $L \subseteq \set{1, \dotsc, r} = V(g)$ of vertices of $g$ and their associated tensors $\set{a(v) : v \in L}$ to a single vertex. Just like we argued above \cref{eq:bounded}, the subset $L$ is also associated with the expectation of some quantum circuit with bounded observables. Therefore, \cref{eq:abound} extends to any subset of vertices:
\begin{equation}
  \abs{a(L)_{s(L)}} \leq 1.
  \label{eq:aL}
\end{equation}
Using this, we can complete the proof by applying \cref{lem:AB} (below)---it bounds by how much a tensor network's value can deviate when each tensor is perturbed entry-wise by some small amount. Indeed, this is precisely what we need to bound $|T(g,\tilde{\a}) - T(g,\a)|$ where $\tilde{\a}$ is our estimate of $\a$.

Let us flesh out this strategy in more detail. Note that the total number of entries in all tensors is upper bounded by $D = r 8^{d'}$, where $r$ is the number of tensors and $d'$ is the maximum degree of the graph $g$. To estimate a single entry of a given tensor, let us use the following parameters:
\begin{align}
  \delta &= \frac{\epsilon}{(e-1)D}, &
  N &= \frac{2\ln(6D)}{\delta^2},
\end{align}
where $\delta$ is the desired accuracy of our estimate, $N$ is the number of times we run the corresponding quantum circuit, and $\epsilon > 0$ can be chosen arbitrarily (it determines the accuracy of our estimate of the tensor network's value).

By Hoeffding's inequality (see above), the probability that the estimate is not within $\delta$ of the actual value is at most
\begin{equation}
  p = 2 \exp \of*{ -\frac{N\delta^2}{2} }.
\end{equation}
By union bound, the probability that at least one out of $D$ estimates if off by more than $\delta$ is at most
\begin{align}
  p D
  &= 2 D \exp \of*{-\frac{N\delta^2}{2}} \\
  &= 2 D \exp \of[\big]{-\ln(6D)} \\
  &= 2 D \cdot \frac{1}{6D}\\
  & = \frac{1}{3}.
\end{align}
In other words, all estimates are correct with probability at least $2/3$, yielding a tensor network $(g,\tilde{\a})$ that is entry-wise $\delta$-close to the original network $(g,\a)$. By \cref{lem:AB} (see below), the values of these tensor networks differ by
\begin{equation}
  |T(g,\tilde{\a}) - T(g,\a)| \leq (e-1) D \delta = \epsilon,
\end{equation}
where $\epsilon$ can be chosen arbitrarily in $(0,1]$.

To analyze the total cost, note that our simulator executes at most $D N$ quantum circuits, where $D$ is an upper bound on the total number of entries of all tensors in the network and $N$ is the number of circuits run to estimate each entry. Expressing this in terms of the relevant parameters,
\begin{align}
  DN
  &= D \cdot 2 \ln (6D) / \delta^2 \\
  &= 2(e-1)^2 D^3 \ln (6D) / \epsilon^2 \\
  &= O \of[\big]{(r^3\log r) 8^{3d'}d'/\epsilon^2} \\
  &= O \of[\big]{(r^3\log r) 2^{O(d')}/\epsilon^2},
\end{align}
where $r = \abs{V(g)}$ and $d'$ is the maximum degree of a vertex in $g$.

To compute the total quantum running time, note that $O(n+m)$ quantum operations are needed for running all $r$ clusters once (by a similar argument as in \cref{thm:main}).
Hence, we can replace a factor of $r$ by $O(n+m)$ and the total quantum running time is bounded by
\begin{align}
    T_Q
    &= \frac{DN}{r}O(n+m)\\
    &= O\of[\big]{(r^2\log r)(n+m) 2^{O(d')}/\epsilon^2}.
\end{align}

After all entries are estimated, we can classically compute the value of the tensor network $(g,\tilde{\a})$ in classical time $T_C = \poly(r)2^{O(\cc(g))}$ by finding a near-optimal tree decomposition of the line graph of $g$~\cite{MS08}. Noting that $d'(g) \leq \cc(g)$, i.e., the maximum degree is at most the contraction complexity for any graph $g$, leads to the simplified formula in the main part of the paper.
\end{proof}

\begin{lemma}\label{lem:AB}
Let $(g,\a)$ be a tensor network with $g = (V,E)$ and $|V| = r$. Assume that, for any subset $L \subseteq V$, all entries of the tensor $a(L)$ obtained by contracting $L$ to a single vertex have norm at most one.
Let $\tilde{\a}$ be another tensor network that is entry-wise close to $\a$, i.e., for all $v \in V$ and $s(v) \in [8]^{d'(v)}$,
\begin{equation}
  |a(v)_{s(v)} - \tilde{a}(v)_{s(v)}| \leq \delta,
\end{equation}
for some $0 \leq \delta \leq 1/(r 8^{d'})$, where $d'(v)$ denotes the degree of a vertex $v$ and $d'=\max_{v} d'(v)$ denotes the maximum degree of $g$. Then
\begin{equation}
  |T(g,\a) - T(g,\tilde{\a})| \leq (e-1) r 8^{d'} \delta.
\end{equation}
\end{lemma}

\begin{proof}
We denote the entry-wise perturbation by $\Delta(v)_{s(v)} = \tilde{a}(v)_{s(v)} - a(v)_{s(v)}$, so $|\Delta(v)_{s(v)}| \leq \delta$. Then, following \cref{eq:TGA},
\begin{align}
 & |T(g,\tilde{\a}) - T(g,\a)| \nonumber \\
 &= \Biggl{|}
      \sum_{s \in [8]^{E}} \prod_{v \in V}
      \of[\big]{ a(v)_{s(v)} + \Delta(v)_{s(v)} } \\ &\;\;
      - \sum_{s \in [8]^{E}} \prod_{v \in V}
      a(v)_{s(v)}
    \Biggr{|}. \nonumber
\end{align}
We expand the first sum, cancel its first term, and parametrize the remaining terms by subsets $L \subsetneq V$ that correspond to those $v \in V$ for which we take the entries of $a(v)$ as opposed to $\Delta(v)$:
\begin{align}
 & |T(g,\tilde{\a}) - T(g,\a)| \\
 &= \abs*{
      \sum_{L \subsetneq V} \sum_{s \in [8]^{E}}
      \of*{
        \prod_{v \in L} a(v)_{s(v)} \cdot
        \prod_{v \in V \setminus L} \Delta(v)_{s(v)}
      }
    }. \nonumber
\end{align}
Let $E(L) = \set{(u,v) \in E : u,v \in L}$. Recall that the value of $a(v)_{s(v)}$ depends only on the edges incident to $v$. Since $v \in V \setminus L$ in the second product, the value of $\Delta(v)_{s(v)}$ is not influenced by the edges in $E(L)$. We can thus decompose the string $s$ into two parts, $s_1$ and $s_2$, where $s_2$ is indexed by the edges in $E(L)$ and $s_1$ is indexed by the remaining edges:
\begin{align}
 & |T(g,\tilde{\a}) - T(g,\a)|
 = \Biggl{|}
     \sum_{L \subsetneq V}
     \sum_{s_1 \in [8]^{E \setminus E(L)}} \\
   & \of[\Bigg]{
       \sum_{s_2 \in [8]^{E(L)}}
       \prod_{v \in L} a(v)_{s(v)}
     } \cdot
     \prod_{v \in V \setminus L} \Delta(v)_{s_1(v)}
   \Biggr{|}. \nonumber
\end{align}
The bracketed expression is equal to $a(L)_{s(L)}$ which are the entries of the tensor obtained by contracting $L$ to a single vertex. Since $\abs{a(L)_{s(L)}} \leq 1$ for all $s$, by triangle inequality we get that
\begin{align}
 & |T(g,\tilde{\a}) - T(g,\a)| \nonumber \\
 & \leq \sum_{L \subsetneq V} \sum_{s_1 \in [8]^{E \setminus E(L)}}
     \abs*{
       \prod_{v \in V \setminus L} \Delta(v)_{s_1(v)}
     } \\
 & \leq \sum_{L \subsetneq V} \sum_{s_1 \in [8]^{E \setminus E(L)}}
     \delta^{|V-L|}.
\end{align}
Since $|E-E(L)| \leq d' \cdot |V-L|$, where $d'$ is the maximum degree,
\begin{align}
  |T(g,\tilde{\a}) - T(g,\a)|
  &\leq \sum_{L \subsetneq V} 8^{d'|V-L|}
        \delta^{|V-L|} \\
  &=    \sum_{L \subsetneq V}
        \of*{ 8^{d'} \delta }^{|V-L|} \\
  &=    \sum_{t=1}^r \binom{r}{t}
        \of*{ 8^{d'} \delta }^t \\
  &= (1+8^{d'}\delta)^{r} - 1.
\end{align}
Note that if $x\leq 1/r$ then
\begin{align}
(1+x)^{r} - 1
&= \sum_{t=1}^{r} \binom{r}{t} x^{t} \label{eq:boundDelta1}\\
&\le \sum_{t=1}^{r} \frac{(rx)^{t}}{t!} \\
&\le (rx)\sum_{t=1}^{r} \frac{1}{t!} \\
&\le (e-1)rx.\label{eq:boundDelta2}
\end{align}
Since $\delta \le 1/(r8^{d'})$,
\begin{align}
|T(g,\tilde{\a}) - T(g,\a)|
&\le (1+8^{d'}\delta)^{r} - 1\\
&\le (e-1)r8^{d'}\delta,
\end{align}
which completes the proof.
\end{proof}

\subsection{Appendix VI: Proof of \cref{thm:Hamiltonian}}

Consider an $N$-qubit system with qubits grouped into $n$ parties.
Let
\begin{equation}
  H = \sum_j H^{(1)}_j + \sum_j H^{(2)}_j, \quad
  \forall i,j: \norm{H_j^{(i)}} \leq 1,
  \label{eq:Ha}
\end{equation}
be a Hamiltonian on this system,
where each term $H^{(i)}_j$ acts on at most two qubits
and $i$ is the number of affected parties
(the $i=1$ terms act within a single party while
the $i=2$ terms act between two different parties).
Denote by $g$ the $n$-vertex graph whose vertices represent the parties
and whose edges correspond to interaction terms $H_{j}^{(2)}$ among different parties.

Starting with a product state $\rho_1 \x \dotsb \x \rho_n$ across all parties
and letting the system evolve for time $t$ with the Hamiltonian $H$,
our goal is to approximate the expectation
\begin{equation}\label{eq:Ex}
  \Tr \sof[\big]{ (O_1 \x \dotsb \x O_n) e^{-iHt} (\rho_1 \otimes \dotsb \otimes \rho_n) e^{iHt} }
\end{equation}
of local observables $O_i$ on each party.

\begin{theorem}\label{thm:Hamiltonian}
Let $H$ be the Hamiltonian in \cref{eq:Ha}
where each qubit is affected by at most a constant number of terms $d'$.
Assume that the states $\rho_i$ and observables $O_i$ in \cref{eq:Ex} can be implemented efficiently,
and that the eigenvalues of each $O_i$ lie within $[-1,1]$.
Then the correlation function in \cref{eq:Ex} can be approximated by an
\begin{align}
 (2^{O\of*{(h_Bt)^2\cc(g) /\epsilon}}, d) \text{-simulator},
\end{align}
where
$h_B = \sum_j \norm{H_j^{(2)}}$
is the total inter-party interaction strength,
$t \geq 0$ is the desired evolution time,
$\cc(g)$ is the contraction complexity of the interaction graph $g$ between the parties,
$\epsilon$ is the desired accuracy,
and $d$ is the number of qubits of the largest party.
\end{theorem}

\begin{proof}
The overall structure of the proof is as follows.
We first implement the evolution of $H$ in the circuit model using a nested Trotter approximation based on \cref{lem:time-step,lem:Hamiltonian} (below).
Then we use \cref{lem:gate-cut} to decompose the circuit and remove the interactions among parties
so that each party can be simulated separately on a small quantum device.

Let us write the Hamiltonian as $H = A + B$
where $A$ and $B$ consist of all $H^{(1)}_j$ and $H^{(2)}_j$ terms, respectively.
Using a nested Trotter approximation, we can decouple $H$
into these two groups and then into individual terms $H^{(i)}_j$.

First, let us decouple $A$ and $B$.
Letting $m_1 = \lceil 4d'h_Bt^2/\epsilon \rceil + \lceil 2h_B^2t^2/\epsilon\rceil$,
\cref{lem:time-step} implies that
\begin{equation}
  \norm*{e^{-iHt} - (e^{-iBt/m_1}e^{-iAt/m_1})^{m_1}} \le 2\epsilon.
\end{equation}
Next, we approximate $e^{-iAt/m_1}$ and $e^{-iBt/m_1}$ by breaking them up into individual terms.
For $e^{-iBt/m_1}$, we directly use \cref{lem:Hamiltonian}:
\begin{equation}
  \norm[\Big]{e^{-iBt/m_1} - \prod_{j} e^{-i H_j^{(2)} t/m_1}} \le 2\epsilon/m_1.
\end{equation}
For $e^{-iAt/m_1}$, we first set
\begin{align}
  m_2
  &= O\of*{h_A^2 (t/m_1)^2 /(\epsilon / m_1)} \\
  &= O(h_A^2 t^2 / (\epsilon m_1))
\end{align}
where $h_A = \sum_{j} \norm{H_{j}^{(1)}}$
and then use \cref{lem:Hamiltonian} and \cref{eq:mUV}:
\begin{equation*}
  \norm[\bigg]{e^{-iAt/m_1}- \of[\Big]{\prod_{j} e^{-i H_j^{(1)} t/m_1m_2}}^{m_2}} \le 2\epsilon/m_1.
\end{equation*}

Combining everything together, the overall Trotter approximation of $e^{-iHt}$ is given by
\begin{equation}
  \tilde{U}
  = \of[\bigg]{\prod_{j_2} e^{-i H_{j_2}^{(2)} t/m_1}
    \of[\Big]{\prod_{j_1} e^{-i H_{j_1}^{(1)} t/m_1m_2}}^{m_2}}^{m_1}
  \label{eq:Ut}
\end{equation}
and its incurred error is
\begin{align}
  &\norm{e^{-iHt} - \tilde{U}} \nonumber \\
  &\le \abs[\big]{1 - \of[\big]{(1+2\epsilon/m_1)(1+2\epsilon/m_1)}^{m_1}} + 2 \epsilon \\
  &= O(\epsilon).
\end{align}

To construct a QC algorithm $(C,f)$ for approximating the correlation function in \cref{eq:Ex}, we need to prepare the states $\rho_i$, approximate the unitary evolution by $H$, and measure the observables $O_i$.
We implement these three steps as follows.

Recall that $N$ is the total number of qubits of all $n$ parties.
Starting from $\ket{0}\xp{N}$, we
first prepare the initial state $\rho = \rho_{1} \otimes \dotsb \otimes \rho_{n}$
by a sequence of single- and two-qubit gates (to account for mixed $\rho_i$, a random sequence of gates can be used.)
By assumption, this can be done efficiently, and we ignore $\poly(N)$ factors in the simulator complexity.

Next, we can approximately implement $e^{-iHt}$ by applying the Trotter unitary $\tilde{U}$ from \cref{eq:Ut}.
Since each factor $e^{-iH_{j}^{(2)}t/m_1}$ and $e^{-iH_{j}^{(1)}t/(m_1m_2)}$ can be implemented by a single- or two-qubit gate,
$\tilde{U}$ consists of $O(Jm_1m_2) = O(Jh_A^2 t^2 /\epsilon) = \poly(N)t^2/\epsilon$ gates, where $J$ is the total number of terms $H_{j}^{(i)}$ in the system (assume $J=\poly(N)$).

Finally, for each observable $O_i$ with spectral decomposition $O_i = \sum_k \alpha^{(i)}_k \proj{\phi_k}$, we can find (again, by assumption) an efficiently implementable unitary $V_i$ such that
\begin{equation*}
  V_i O_i V_i\ct = \sum_k \alpha^{(i)}_k \proj{k}
\end{equation*}
where $\ket{k}$ is the standard basis. Let $f_i(k) = \alpha^{(i)}_k$. To implement $O = O_1 \otimes \dotsb \otimes O_{n}$, we can first implement the basis change $V_i$ for each party, measure each party in the standard basis to obtain outcomes $y_1,\ldots,y_n$, and then output $f(y) = \prod_{i=1}^n f_i(y_i)$.

This procedure amounts to a QC algorithm $(C,f)$ whose output expectation is
\begin{align*}
  \E_{y}f(y)
  = \Tr \sof[\big]{
      (O_1 \x \dotsb \x O_n) \tilde{U}
      (\rho_1 \otimes \dotsb \otimes \rho_n) \tilde{U}\ct
    },
\end{align*}
where $\tilde{U}$ is a Trotter approximation of $e^{-iHt}$.
To bound the error from replacing $e^{-iHt}$ with $\tilde{U}$,
observe that
\begin{align}
&\abs[\big]{\Tr[\rho e^{iHt}Oe^{-iHt}] -
 \Tr [\rho \tilde{U}^{\dagger}O\tilde{U}]} \nonumber \\
&\leq \norm{e^{iHt}Oe^{-iHt} - \tilde{U}^{\dagger}O\tilde{U}} \\
&\leq \norm{ (e^{-iHt}-\tilde{U})^{\dagger} O e^{-iHt}} \\
&\quad+ \norm{ \tilde{U}^{\dagger}O (e^{-iHt} - \tilde{U})}\nonumber\\
&\leq \norm{e^{-iHt}-\tilde{U}} + \norm{e^{-iHt} - \tilde{U}}\\
&\leq O(\epsilon),
\end{align}
where the second-to-last inequality uses the assumption that $\|O\|\leq 1$.

The final step is to evoke \cref{lem:gate-cut} (below).
Partitioning the qubits in the QC algorithm $(C,f)$ according to the parties in the Hamiltonian simulation, each party has at most $d$ qubits.
We denote by $g'$ the $n$-vertex graph whose vertices represent the parties and whose edges correspond to the two-qubit gates between different parties in the circuit $C$.
Denote by $d_{g'}(i,j)$ the number of edges between vertices $i$ and $j$ in $g'$.
Since each $H_{j}^{(2)}$ contributes $O(m_1) = O(h_{B}^2t^2/\epsilon)$ gates to the Trotter approximation $\tilde{U}$,
\begin{equation}
  d_{g'}(i,j) = O(h_{B}^2t^2/\epsilon) \cdot d_{g}(i,j),
  \label{eq:dd}
\end{equation}
where $g$ is the interaction graph of the Hamiltonian $H$ and
$d_{g}(i,j)$ is the number of edges between vertices $i$ and $j$ in $g$.
Hence, the contraction complexities of $g'$ and $g$ are related as follows:
\begin{equation}
  \cc(g') = O(h_{B}^{2}t^{2}\cc(g)/\epsilon).
\end{equation}
By \cref{lem:gate-cut}, we conclude that a
\begin{align}
 (2^{O((h_Bt)^2\cc(g) /\epsilon)}, d) \text{-simulator}
\end{align}
exists for the Hamiltonian $H$.
\end{proof}

\begin{lemma}[Lemma 4 in \cite{haah2018quantum}]\label{lem:interaction}
Let $A_t$ and $B_t$ be continuous time-dependent Hermitian operators, and let $U_t^{A}$ and $U_{t}^{B}$ with $U_0^{A} = U_{0}^{B} = I$ be the corresponding time evolution unitaries. Then the following hold:
\begin{itemize}
\item[(a)] $W_t = {U_t^B}\ct U_t^{A}$ is the unique solution of\\$i \frac{dW_t}{dt} = \of[\big]{{U_t^B}\ct (A_t - B_t) U_t^{B}} W_t$ and $W_0 = I$.
\item[(b)] If $\norm{A_s - B_s} \le \delta$ for all $s\in [0,t]$, then\\$\norm{U_t^{A} - U_t^{B}} \le t\delta$.
\end{itemize}
\end{lemma}

\newcommand{\ad}{\mathrm{ad}}

\begin{lemma} \label{lem:commutator}
For Hermitian matrices $A$ and $B$,
\begin{equation*}
\norm*{e^{-i(A+B)t} - e^{-iBt}e^{-iAt}} \le \sum_{k=1}^{\infty} \frac{t^{k+1}}{k!} \norm{\ad_{B}^{k}(A)}
\end{equation*}
where $\ad_B$ is the adjoint map $\ad_{B} (A) = [B,A]$ and $\ad_{B}^{k}$ is the $k$-th power of $\ad_{B}$:
\begin{equation}
  \ad_{B}^{k}(A) = [\underbrace{B, [\ldots [B, [B}_{k}, A]]\ldots]].
\end{equation}
\end{lemma}
\begin{proof}
Note that $e^{-i(A+B)t}$ is generated by the Hamiltonian $A+B$ and
$e^{-iBt}e^{-iAt}$ is generated by the Hamiltonian $e^{-iBt}(A+B)e^{iBt}$
by part~(a) of \cref{lem:interaction}.
According to Lemma~5.3 in~\cite{miller1973symmetry},
\begin{equation}
e^{X}Ye^{-X} = \sum_{k=0}^{\infty} \frac{1}{k!}\ad_X^{k}(Y).
\end{equation}
Using this identity we get, for any $s \in [0,t]$,
\begin{align}
&\norm*{e^{-iB s}(A+B) e^{iB s} - (A+B)} \nonumber\\
&= \norm*{\sum_{k=0}^{\infty} \ad_B^{k}(A+B)\frac{(-is)^{k}}{k!} - (A+B)}\\
&= \norm*{\sum_{k=1}^{\infty} \ad_B^{k}(A)\frac{(-is)^{k}}{k!}}\\
& \le \sum_{k=1}^{\infty} \frac{s^{k}}{k!}\norm{\ad_{B}^{k}(A)}\\
& \le \sum_{k=1}^{\infty} \frac{t^{k}}{k!}\norm{\ad_{B}^{k}(A)}.
\end{align}
By part (b) of \cref{lem:interaction},
\begin{equation*}
\norm*{e^{-i(A+B)t} - e^{-iBt}e^{-iAt}}
\le t \sum_{k=1}^{\infty} \frac{t^{k}}{k!} \norm{\ad_{B}^{k}(A)},
\end{equation*}
which is the desired result.
\end{proof}

\begin{lemma}\label{lem:lowdegree}
Let $H = A + B = \sum_{j} H^{(1)}_j + \sum_{j} H^{(2)}_j$
where each $H_j$ is a one- or two-qubit Hamiltonian with $\norm{H_j} \le 1$.
If each qubit is involved in at most $d'$ terms
and $h_B = \sum_j \norm{H_j^{(2)}}$ denotes the total interaction strength, then
\begin{equation}
  \norm{\ad_B^{k}(A)} \le 2^{k+1} d' h_{B}^{k}.
\end{equation}
\end{lemma}

\begin{proof}
Since $\norm{[X,Y]} \le 2\norm{X}\norm{Y}$ for any matrices $X$ and $Y$,
\begin{align}
  \norm{\ad_{B}^{k}(A)}
  &\leq 2 \norm{B} \norm{\ad_{B}^{k-1}(A)} \\
  &\leq (2\norm{B})^{k-1}\norm{[A, B]} \\
  &\leq (2h_B)^{k-1}\norm{[A,B]}.
\end{align}
To bound $\norm{[A,B]}$, note that
\begin{align}
\norm{[A,B]}
&\le \sum_{j} \norm{[A,H_j^{(2)}]}\\
&\le \sum_{j} \sum_i \norm{[H_i^{(1)}, H_j^{(2)}]}\\
&\le \sum_{j} \sum_{i: [H_i^{(1)}, H_j^{(2)}] \neq 0} 2 \norm{H_i^{(1)}} \norm{H_j^{(2)}}\\
&=\sum_{j} \norm{H_j^{(2)}} \sum_{i:[H_i^{(1)}, H_j^{(2)}] \neq 0} 2.
\end{align}
For a fixed $j$, assume $H_{j}^{(2)}$ acts on qubits $a$ and $b$.
Then $[H_i^{(1)}, H_j^{(2)}]$ can be non-zero only when $H_i^{(1)}$ also acts on $a$ or $b$.
The number of such $H_i^{(1)}$ that do not commute with $H_j^{(2)}$ is bounded by $2d'$.
Hence
\begin{align}
\norm{[A,B]}
&\le \sum_{j} \norm{H_{j}^{(2)}} \, 4d'\\
&\le 4h_Bd'.
\end{align}
Combining everything together we get
$\norm{\ad_{B}^{k}(A)} \le (2h_B)^{k-1} \cdot 4h_Bd' = 2^{k+1} d' h_B^{k}$.
\end{proof}

\begin{lemma}\label{lem:time-step}
Assume the same setting as in \cref{lem:lowdegree}.
Then, for any positive $\epsilon \le d't$,
\begin{equation}
  \norm*{e^{-i(A+B)t} - (e^{-iBt/m}e^{-iAt/m})^{m}} \le 2 \epsilon
\end{equation}
whenever $m \ge \lceil{4d' h_B t^2 / \epsilon\rceil}$.
\end{lemma}

\begin{proof}
First, by combining \cref{lem:commutator,lem:lowdegree},
\begin{align}
&\norm*{e^{-i(A+B)t/m} - (e^{-iBt/m}e^{-iAt/m})}\nonumber\\
&\le \sum_{k=1}^{\infty} \frac{t^{k+1}}{m^{k+1}k!} \cdot (2^{k+1} d' h_B^{k})\\
&\le \frac{1}{m} \sum_{k=1}^{\infty}
\frac{2^{k+1}d'h_B^{k}}{\of*{4d'h_Bt^2/\epsilon}^k}
\frac{t^{k+1}}{k!}\\
& = \frac{1}{m} \sum_{k=1}^{\infty} \frac{\epsilon^{k}}{(2d't)^{k-1}k!} \\
& \leq \frac{\epsilon}{m} \sum_{k=1}^{\infty} \of*{\frac{\epsilon}{2d't}}^{k-1}\\
& \le \frac{2\epsilon}{m},
\end{align}
where the last inequality follows from $\epsilon \le d't$ and $1+\frac{1}{2}+\frac{1}{4} + \dotsc = 2$.
Note that
$U^m-V^m = U^{m-1}(U-V) + U^{m-2}(U-V)V + \dotsb + (U-V)V^{m-1}$
so, for any unitary operations $U$ and $V$,
\begin{equation}
  \norm{U^m-V^m} \le m \norm{U-V} \label{eq:mUV},
\end{equation}
where the last inequality follows from the triangle inequality and $\norm{U} = \norm{V} = 1$. Using this,
\begin{equation}
  \norm*{e^{-i(A+B)t} - (e^{-iBt/m}e^{-iAt/m})^{m}}
  \le m \cdot \frac{2\epsilon}{m},
\end{equation}
which is the desired bound.
\end{proof}

\begin{lemma} \label{lem:Hamiltonian}
Let $H = \sum_{j} H_{j}$ and $h=\sum_{j}\norm{H_j}$.
For any positive $\epsilon \le ht$,
\begin{equation}
  \norm[\Big]{e^{-iHt/m} - \prod_j e^{-iH_jt/m}} \le 2\epsilon/m
\end{equation}
whenever $m\ge \lceil 2h^2t^2/\epsilon \rceil$.
\end{lemma}

\begin{proof}
Let $L$ denote the number of terms in $H$.
By direction expansion,
\begin{align}
&\norm[\Big]{e^{-iHt/m} - \prod_{j} e^{-iH_jt/m}} \nonumber \\
&= \norm[\Bigg]{\sum_{k=0}^{\infty} (H_1+\dotsb+H_L)^{k} (t/m)^{k} / k! \\
&\quad-
  \sum_{k=0}^{\infty} \sum_{x_1+\dotsb+x_L=k}
  \frac{H_1^{x_1}\dotsb H_L^{x_L}(t/m)^{k}}{x_1!\dotsb x_L!}} \nonumber \\
&=
  \norm[\Bigg]{\sum_{k=2}^{\infty} (H_1+\dotsb+H_L)^{k} (t/m)^{k} / k! \\
&\quad-
  \sum_{k=2}^{\infty} \sum_{x_1+\dotsb+x_L=k}
  \frac{H_1^{x_1}\dotsb H_L^{x_L}(t/m)^{k}}{x_1!\dotsb x_L!}} \nonumber \\
&\le
  \sum_{k=2}^{\infty} \norm*{H_1+\dotsb+H_L}^k (t/m)^{k} /k! \\
&\quad+
  \sum_{k=2}^{\infty} \sum_{x_1+\dotsb+x_L=k}
  \frac{\norm{H_1^{x_1}\dotsb H_L^{x_L}}(t/m)^{k}}{x_1!\dotsb x_L!} \nonumber \\
&\le
  \sum_{k=2}^{\infty} (ht/m)^{k}/k! + \sum_{k=2}^{\infty} (ht/m)^{k}/k! \\
&\le
  \frac{1}{m}2\sum_{k=2}^{\infty} \frac{(ht)^{k}}{m^{k-1}} \frac{1}{k!}.
\end{align}
If $m\geq \lceil 2h^2t^2/\epsilon \rceil$,
\begin{align}
&\norm[\Big]{e^{-iHt/m} - \prod_{j} e^{-iH_jt/m}}\nonumber\\
&\le \frac{1}{m} 2\sum_{k=1}^{\infty} \frac{(ht)^{k+1}}{(2h^2t^2/\epsilon)^k k!} \\
&= \frac{1}{m} 2\sum_{k=1}^{\infty} \frac{\epsilon^{k}}{2^{k}h^{k-1}t^{k-1} k!}\\
&=\frac{\epsilon}{m} \sum_{k=1}^{\infty} \of*{\frac{\epsilon}{2ht}}^{k-1} \frac{1}{k!}\\
& \le \frac{2\epsilon}{m},
\end{align}
as desired.
\end{proof}

\newcommand{\Qart}{\mathcal{Q}}

\begin{lemma}\label{lem:gate-cut}
Let $(C,f)$ be a QC algorithm and $\Qart = \set{Q_1,\dotsc,Q_r}$ be a partition of its qubits.
Let $g'$ be the $r$-vertex (multi-)graph obtained by representing each subset $Q_i$ as a vertex $i$ and
each two-qubit gate in $C$ acting across sets $Q_i$ and $Q_j$ as an edge $(i,j)$.
Let $K$ be the number of edges in $g'$ and let $d=\max_{i}|Q_i|$ be the size of the largest set of qubits.
Then $(C,f)$ has a $(2^{O(K)},d)$-simulator.
If in addition $f$ is decomposable,
i.e., $f(y) = \prod_{i=1}^r f_i(y_i)$ where $y_i$ is the outcome of the standard basis measurement on qubits $Q_i$,
and $f_i(y_{i}) \in [-1,1]$ then $(C,f)$ has a $(2^{O(\cc(g'))}, d)$-simulator.
\end{lemma}

\begin{proof}
Let $(G,\A)$ be the tensor network of $(C,f)$.
Our strategy is to construct a clustering of $G$ whose induced graph $g$ is very similar to the graph $g'$ given in the Lemma.
The two simulators are then obtained by applying \cref{thm:main,thm:Treewidth}, respectively.

We obtain the desired clustering of $G$ from the graph $g'$.
Since each cluster corresponds to either a vertex or an edge in $g'$, we get two types of clusters.
Let us describe them in more detail.

For each vertex $i$ of $g'$, we introduce a cluster $S_i$ that contains tensors that act only on qubits $Q_i$.
That is, $S_i$ consists of all vertices of $G$ that correspond to either an initial qubit state or a single- or a two-qubit gate in the set $Q_i$.

Similarly, we introduce a separate cluster also for each edge of $g'$
(recall that the edges of $g'$ correspond to two-qubit gates between different parts of $\Qart$).
We denote these clusters by $T_j$.

Since the number of edges of $g'$ is $K$, the overall clustering of $G$ is given by
\begin{equation}
  \Part = \set{S_1, \dotsc, S_r}
     \cup \set{T_1, \dotsc, T_K},
\end{equation}
where (as usual) the final observable $O_f$ has been left out of the clustering.
This partition induces a graph $g$ \footnote{Although we reuse the symbol $g$, here it has a slightly different meaning than in \cref{thm:Hamiltonian} where this Lemma is evoked. Here the graph $g$ is induced from a circuit $C$ while in \cref{thm:Hamiltonian} it is induced by interactions of a Hamiltonian $H$. The former is obtained from the latter via a Trotter approximation, see \cref{eq:dd}.} that is very similar to $g'$.
In particular, the number of edges in $g$
(or the number of qubits communicated between clusters of $\Part$)
is $O(K)$.

Let us verify that the number of qubits needed for simulating each cluster of $\Part$ is also $d$.
Following the strategy of \cref{thm:main}, we can repeatedly apply \cref{lem:edgecut} to cut the tensor network into separate clusters.
Since all interactions between qubits in different parts of $\Qart$ are mediated by two-qubit gates $U_j$ represented by the clusters $T_j$,
splitting the tensor network into clusters $\Part$ effectively amounts to cutting these two-qubit gates $U_j$ out of the circuit.
For each $U_j$, we apply \cref{lem:edgecut} four times:
\begin{center}
\begin{tikzpicture}[thick, > = latex,
  gate/.style = {fill = white, draw, text height = 1.5ex, text depth = .25ex, minimum size = 30pt},
  pt/.style = {circle, draw = black, fill = black, inner sep = 1.5pt},
  tr/.style = {isosceles triangle, isosceles triangle stretches, inner sep = 0pt, minimum height = 7pt, minimum width = 11pt, fill = white},
  ob/.style = {tr, draw},
  st/.style = {tr, draw, shape border rotate = 180},
  bd/.style = {draw = none, rounded corners = 2pt},
	nr/.style = {draw = none, rounded corners = 0pt},
	lo/.style = {orange!20},
	lb/.style = {blue!20}
]

  \def\h{0.30}  
  \def\H{0.7}   
  \def\e{0.06}  
  \def\c{0.6}   

  \def\UL{-4.5} 
  \def\UR{-2.7} 

  \fill[lb]
  (\UL-0.3,\H+\c) --
  (\UL-0.3,\e)[bd] --
  (\UL+0.1,\e) --
  (\UL+0.1,\H) --
  (\UR-0.3,\H) --
  (\UR-0.3,\e)[nr] --
  (\UR+0.1,\e) --
  (\UR+0.1,\H+\c) -- cycle;

  \fill[lo]
  (\UL-0.3,-\H-\c) --
  (\UL-0.3,-\e)[bd] --
  (\UL+0.1,-\e) --
  (\UL+0.1,-\H) --
  (\UR-0.3,-\H) --
  (\UR-0.3,-\e)[nr] --
  (\UR+0.1,-\e) --
  (\UR+0.1,-\H-\c) -- cycle;

  \draw[->] (\UL, \h) -- (\UR, \h);
  \draw[->] (\UL,-\h) -- (\UR,-\h);

  \node[gate] (I) at (-3.7,0) {$U_j$};

  \node at (-2.3,0) {$=$};
  \node at (-1.6,0) {$\displaystyle\sum_{k_1,\dotsc,k_4=1}^8 \!\!\! a_k$};

  \def\UL{0.0} 
  \def\UR{2.4} 
  \def\gp{0.65} 

  \fill[lb]
  (\UL-\gp,\H+\c) --
  (\UL-\gp,\e)[bd] --
  (\UL+0.0,\e) --
  (\UL+0.0,\H) --
  (\UR+0.0,\H) --
  (\UR+0.0,\e)[nr] --
  (\UR+\gp,\e) --
  (\UR+\gp,\H+\c) -- cycle;

  \fill[lo]
  (\UL-\gp,-\H-\c) --
  (\UL-\gp,-\e)[bd] --
  (\UL+0.0,-\e) --
  (\UL+0.0,-\H) --
  (\UR+0.0,-\H) --
  (\UR+0.0,-\e)[nr] --
  (\UR+\gp,-\e) --
  (\UR+\gp,-\H-\c) -- cycle;

  \node[ob] (O1) at (\UL-0.15, \h) {}; \draw[->] (O1)+(-0.45,0) -- (O1);
  \node[ob] (O2) at (\UL-0.15,-\h) {}; \draw[->] (O2)+(-0.45,0) -- (O2);
  \node[st] (R1) at (\UL+0.15, \h) {};
  \node[st] (R2) at (\UL+0.15,-\h) {};

  \def\d{0.60} 

  \path (O1)+(-0.05, \d) node {\scriptsize$O_{k_1}^1$};
  \path (O2)+(-0.05,-\d) node {\scriptsize$O_{k_2}^2$};
  \path (R1)+( 0.23, \d) node {\scriptsize$\rho_{k_1}^1$};
  \path (R2)+( 0.23,-\d) node {\scriptsize$\rho_{k_2}^2$};

  \node[ob] (O3) at (\UR-0.15, \h) {}; \draw[->] (R1) -- (O3);
  \node[ob] (O4) at (\UR-0.15,-\h) {}; \draw[->] (R2) -- (O4);
  \node[st] (R3) at (\UR+0.15, \h) {}; \draw (R3) -- ++(0.4,0);
  \node[st] (R4) at (\UR+0.15,-\h) {}; \draw (R4) -- ++(0.4,0);

  \path (O3)+(-0.05, \d) node {\scriptsize$O_{k_3}^3$};
  \path (O4)+(-0.05,-\d) node {\scriptsize$O_{k_4}^4$};
  \path (R3)+( 0.23, \d) node {\scriptsize$\rho_{k_3}^3$};
  \path (R4)+( 0.23,-\d) node {\scriptsize$\rho_{k_4}^4$};

  \node[gate] (I) at (0.5*\UL+0.5*\UR,0) {$U_j$};

\end{tikzpicture}
\end{center}
where $a_k = c_{k_1} c_{k_2} c_{k_3} c_{k_4} \in \{-\frac{1}{2^{4}}, \frac{1}{2^{4}}\}$
is the product of the four coefficients from \cref{lem:edgecut}
and $O_{k_1}^1, \dotsc, O_{k_4}^4$
and $\rho_{k_1}^1, \dotsc, \rho_{k_4}^4$
are the four associated observables and states, respectively.
Note that we can simulate each $U_j$ classically since only two qubits are involved.
This takes care of all clusters $T_j$.

The remaining clusters $S_i$ are associated to subsets of qubits $Q_i$.
Let us argue that $|Q_i|$ qubits are enough for simulating $S_i$ if qubit recycling is permitted.
Indeed, for each input qubit of $S_i$ that results from cutting out some two-qubit gate $U_j$
(e.g., $\rho_{k_3}^3$ or $\rho_{k_4}^4$ in the above figure),
there always exists an output qubit of $S_i$
(e.g., the qubit measured by $O_{k_1}^{1}$ or $O_{k_2}^{2}$ in the above figure)
that can be recycled to be the new input qubit after the measurement.
Hence, no additional qubits -- other than the original $|Q_i|$ ones -- are needed for simulating the cluster $S_i$.

To summarize, the original parameters $K$ and $d$ of the graph $g'$ defined in the statement of the Lemma translate into $O(K)$ and $d$ for the graph $g$ induced by the partition $\Part$ of $G$.
Hence, the circuit $C$ is $(O(K),d)$-clustered
and thus by \cref{thm:main} has a $(2^{O(K)},d)$-simulator.

If $f$ is decomposable, we can get a $(2^{O(\cc(g))}, d)$-simulator by using \cref{thm:Treewidth}.
We can find a good contraction order for $g$ based on $g'$.
Indeed, we can first absorb each cluster $T_j$ that represents some two-qubit gate $U_j$
into one of its adjacent clusters $S_i$ or $S_{i'}$ --
this effectively introduces a double edge between $S_i$ and $S_{i'}$.
Repeating this for each $T_j$ reduces the graph $g$ to $g'$.
Hence $\cc(g) = O(\cc(g'))$, yielding a $(2^{O(\cc(g'))},d)$-simulator for $(C,f)$.
\end{proof}

\subsection{Appendix VII: VQE experiments}

We consider VQE algorithms using the following parameterized circuit with $n$ qubits and $D$ layers proposed in \cite{kandala2017hardware}:
\begin{equation}
  U(\theta) = U_D(\theta_D) \Uent \dotsb U_1(\theta_1) \Uent U_0(\theta_0),
  \label{eq:Utheta}
\end{equation}
where $U_i(\theta_i) = \bigotimes_{j=1}^n U_i^j(\theta_i^j)$ and each $U_i^j(\theta_i^j)$ is a parameterized single-qubit gate applied on the $j$-th qubit:
\begin{equation}
  U_i^j(\theta_i^j) = Z_j(\theta_{i,1}^j) X_j(\theta_{i,2}^j) Z_j(\theta_{i,3}^j),
\end{equation}
where $Z_j(\beta) = \exp(-i \beta Z)$ and $X_j(\beta) = \exp(-i \beta X)$ are single-qubit $X$ and $Z$ rotations applied to the $j$-th qubit. Since $U_0(\theta_0)$ is applied directly to the input $\ket{0}^{n}$, the first $Z$ rotation gates can be removed:
\begin{equation}
  U_0(\theta_0) = \bigotimes_{j=1}^{n} Z_j(\theta_{0,1}^j) X_j(\theta_{0,2}^j),
\end{equation}
so in total $\theta$ has $(3D+2)n$ parameters. Each layer of single-qubit gates in \cref{eq:Utheta} is followed by the following layer of entangling gates:
\begin{align}
  \Uent = \prod_{i=1}^{n-1} \mathrm{CZ}(i,i+1),
\end{align}
where $\mathrm{CZ}(i,j)$ denotes the controlled $Z$ gate between qubits $i$ and $j$, defined as $\mathrm{CZ} \ket{a,b} = (-1)^{ab} \ket{a,b}$, for all $a,b \in \set{0,1}$.

We run the VQE circuit from~\cite{kandala2017hardware} with $n=6$ and $D=1$ on the 5-qubit ``ibmq\_ourense'' cloud quantum computer provided by the IBM Quantum Experience \cite{IBM:20}. Our experiment uses up to three qubits. Our goal is to solve the optimization problem $\min_\theta F(\theta)$ with
\begin{align}
F(\theta) = \bra{0}\xp{6} U(\theta)^{\dag} H U(\theta) \ket{0}\xp{6}
\end{align}
where $\theta$ has $(3+2)\cdot6=30$ parameters and the minimal eigenvalue of $H$ approximates the ground energy of $\text{BeH}_{2}$ with an interatomic distance of 1.7~\AA~\footnote{We choose $H$ as a diagonal matrix, which both provides low measurement complexity and desired approximation precision of the energy.}.

We optimize the parameters of $\theta$ using the simultaneous perturbation stochastic approximation (SPSA) method. Let $\theta(k)$ denote the parameters at the $k$-th iteration of SPSA. In each iteration, a random vector $r(k)$ is generated, where each element of $r(k)$ is drawn uniformly from $\{1, -1\}.$ Let $\theta^{\pm}(k) = \theta(k) \pm c_{k} r(k)$, where $c_{k}$ are hyper-parameters. The gradient of $\theta(k)$ is estimated by
\begin{align}
g(k) = \frac{F(\theta^{+}(k)) - F(\theta^{-}(k))}{2 c_k r(k)}
\end{align}
where $F$ is evaluated by the quantum device. Then, $\theta$ is updated as follows:
\begin{align}
\theta(k+1) = \theta(k) - a_k g(k)
\end{align}
where $a_{k}$ controls the learning rate. We set $c_{k} = 0.3/\sqrt{k}$ and $a_{k} = k^{-0.3}$, and perform $200$ iterations.

In \cite{kandala2017hardware}, $F(\theta)$ is evaluated by running $U(\theta)$ on a $6$-qubit quantum computer. Using our cluster simulation scheme, we evaluate $F$ using only $3$~qubits. To do so, we partition the six qubits of $U(\theta)$ into two sets: $\{1,2,3\}$ and $\{4,5,6\}.$ We then decompose the only two-qubit gate CZ$(3,4)$ that acts between these two sets using the gadget proposed in~\cite{mitarai2019constructing} (see \cref{fig:CZ-decomposition}~\footnote{We have corrected typos that appear in Fig.~2 of~\cite{mitarai2019constructing}.}). Note that $Z$, $RZ(\pi/2)$, and $(I+aZ)/2$ are all single-qubit operations~\footnote{To implement $(I+aZ)/2$ on the IBM Quantum Experience, we directly measure the qubit in the $Z$-basis and do classical post-processing.}, so the circuits on $\{1,2,3\}$ and $\{4,5,6\}$ can be run separately using only $3$~qubits. For simulating each run of $U(\theta)$, we need to execute $12$ different circuits on a $3$-qubit computer. For each of these circuits, we produce $8000$ samples to estimate their output. The value of $F(\theta)$ is then obtained by aggregating and post-processing these outputs.

The resulting accuracy of our estimate is similar to the one obtained using the full $6$-qubit circuit~\cite{kandala2017hardware}, including the evaluation of $\theta^{-}(k), \theta^{+}(k)$ and $\theta(k)$. See \cref{fig:VQE-details}. This demonstrates the potential of our cluster simulation scheme for decreasing the number of qubits required in near-term VQE applications. Our source code can be found at:
\url{https://github.com/TianyiPeng/Partiton_VQE}.
\begin{figure}
\includegraphics[scale=0.45]{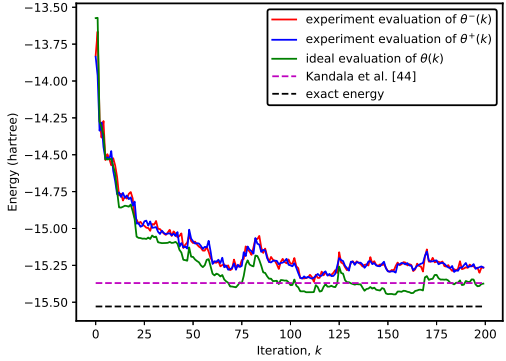}
\caption{Experiment results for estimating the ground energy of $\text{BeH}_2$ by running the 6-qubit VQE of \cite{kandala2017hardware} on a 3-qubit device. The red (blue) line is $F(\theta^{-}(k))$ ($F(\theta^{+}(k))$) estimated in the experiments at the $k$-th step. The green line is the value of $F(\theta(k))$ evaluated ideally. The pink dotted line is the final energy estimate from \cite{kandala2017hardware} with sufficient samples to approximate an ideal evalution. The black dotted line is the exact ground energy of $\text{BeH}_2$ with the with interatomic distance of 1.7~\AA{}.
}
\label{fig:VQE-details}
\end{figure}

\usetikzlibrary{decorations.pathreplacing}
\begin{figure*}
\begin{center}
\begin{tikzpicture}[thick, > = latex,
  gate/.style = {fill = white, draw, text height = 1.5ex, text depth = .25ex, minimum size = 40pt},
  sgate/.style = {fill = white, draw, text height = 1.5ex, text depth = .25ex, minimum size = 15pt},
  pt/.style = {circle, draw = black, fill = black, inner sep = 1.5pt},
  tr/.style = {isosceles triangle, isosceles triangle stretches, inner sep = 0pt, minimum height = 7pt, minimum width = 11pt, fill = white},
  ob/.style = {tr, draw},
  st/.style = {tr, draw, shape border rotate = 180},
  bd/.style = {draw = none, rounded corners = 2pt},
	nr/.style = {draw = none, rounded corners = 0pt},
	lo/.style = {orange!20},
	lb/.style = {blue!20}
]

  \def\h{0.50}  
  \def\H{0.7}   
  \def\e{0.06}  
  \def\c{0.6}   

  \def\UL{-5} 
  \def\UR{-1.0} 

  \node at (-5.5, 0) {$2$};
  \draw[->] (\UL, \h) -- (\UR, \h);
  \draw[->] (\UL,-\h) -- (\UR,-\h);

  \node[gate] (I) at (-3.7,0) {$CZ$};

  \node[sgate] at (-2, \h) {$RZ(\frac{\pi}{2})$};
  \node[sgate] at (-2, -\h) {$RZ(\frac{\pi}{2})$};

  \node at (-0.5,0) {$=$};

  \def\UL{0.0} 
  \def\UR{0.9} 

  \draw[->] (\UL, \h) -- (\UR, \h);
  \draw[->] (\UL,-\h) -- (\UR,-\h);

  \node at (1.2, 0) {$+$};
  \def\UL{1.5} 
  \def\UR{2.6} 

  \draw[->] (\UL, \h) -- (\UR, \h);
  \draw[->] (\UL,-\h) -- (\UR,-\h);
  \node[sgate] at (2, \h) {$Z$};
   \node[sgate] at (2, -\h) {$Z$};

  \node at (3.0, 0) {$+$};
  \node at (4.2, -0.1) {$\displaystyle \sum_{a_1,a_2 \in \{\pm 1\}^2} \hspace{-5mm} a_1 a_2$};
  \def\UL{5.7} 
  \def\UR{7.9} 

  \draw[->] (\UL, \h) -- (\UR, \h);
  \draw[->] (\UL,-\h) -- (\UR,-\h);

  \node[sgate] at (\UL+1, \h) {$\frac{I+a_2Z}{2}$};
  \node[sgate] at (\UL+1, -\h) {$RZ(\frac{a_1\pi}{2})$};

  \draw [decorate,decoration={brace,amplitude=2ex},xshift=-0.5ex,yshift=0pt]
(5.7,-\h-0.6) -- (5.7,\h+0.6);

  \node at (8.2, 0) {$+$};
  \def\UL{8.4} 
  \def\UR{10.6} 

  \draw[->] (\UL, \h) -- (\UR, \h);
  \draw[->] (\UL,-\h) -- (\UR,-\h);

  \node[sgate] at (\UL+1, \h) {$RZ(\frac{a_1\pi}{2})$};
  \node[sgate] at (\UL+1, -\h) {$\frac{I+a_2Z}{2}$};

  \draw [decorate,decoration={brace,amplitude=2ex,mirror},xshift=-0.5ex,yshift=0pt]
(10.8,-\h-0.6) -- (10.8,\h+0.6);

\end{tikzpicture}
\end{center}
\caption{Decomposition of the $CZ$ gate where $RZ(\theta) = e^{-i\theta Z/2}$ \cite{mitarai2019constructing}.}
\label{fig:CZ-decomposition}
\end{figure*}
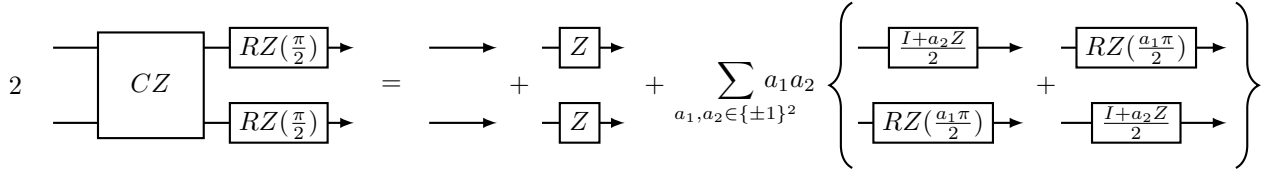

\subsection{Appendix VIII: Reducing $\cc(g)$ in VQE}
To reduce the number of circuits required to run, we propose the following steps when applying our scheme to the VQE circuit in \cref{eq:Utheta}
\begin{itemize}
\item[1.] choose a partition $\Part = \set{S_1, \dotsc, S_r}$ of $n$ qubits such that $\abs{S_i} \leq d$ for each $i$;
\item[2.] remove some entangling gates from $\Uent$ that go across different parts of $\Part$ and denote the modified sequence by $\Uent'$;
\item[3.] replace some occurrences of $\Uent$ in the original circuit $U(\theta)$ by $\Uent'$ and denote the modified circuit by $U'(\theta)$.
\end{itemize}
Given the circuit $U'(\theta)$ with fixed parameters $\theta$, we estimate $\bra{0}\xp{n} U'^{\dag}(\theta) H U'(\theta) \ket{0}\xp{n}$ by performing Pauli measurements and summing according to the Pauli decomposition of $H$:
\begin{equation}
H = \sum_{i=1}^{L} \alpha_i \of[\bigg]{\bigotimes_{j=1}^{n}\sigma_{i,j}},
\end{equation}
where $\sigma_{i,j} \in \{X,Y,Z,I\}$ is a Pauli matrix. Since the measurement is Pauli observable, we can use a QC circuit $(C, f)$ to compute the expectation value, where $f$ is decomposable. Denote $g$ by regarding each set $S_i$ as a node and each gate that acts cross two sets as an edge. By \cref{lem:gate-cut}, the modified circuit has a $(2^{O(\cc(g))},d)$-simulator where $\cc(g)$ is the contraction complexity of the graph $g$.

This leads to a trade-off in time versus performance: replacing $\Uent$ by $\Uent'$ decreases $\cc(g)$ and hence reduces the simulation time; however, this also decreases the expressive power of the circuit and thus degrades the quality of the solution. This may be justified in physically meaningful cases, e.g.,when the true ground state obeys an area law~\cite{RevModPhys.82.277}.

Our numerical observations suggest that for a small scale quantum circuit, replacing some $\Uent$ by $\Uent'$ does not influence the performance significantly. Specifically, we let $\Uent$ be a fixed sequence of two-qubits gates for producing entanglement:
\begin{equation}
  \Uent = \prod_{i=1}^{n-1} \prod_{j=i+1}^{n} \mathrm{CNOT}(i,j),
\end{equation}
where $\mathrm{CNOT}(i,j)$ is the CNOT gate with the $i$-th qubit as control and the $j$-th qubit as target.

We partition our $n$ qubits into two sets: $S_1 = \set{1,\dotsc,\frac{n}{2}}$ and $S_2 = \set{\frac{n}{2}+1,\dotsc,n}$ and construct $\Uent'$ form $\Uent$ by removing CNOT gates that cross the two sets:
\begin{align}
\Uent' =& \left(\prod_{i=n/2+1}^{n-1}\prod_{j=i+1}^{n} \mathrm{CNOT}(i,j) \right)\\ &\cdot  \left(\prod_{i=1}^{n/2-1}\prod_{j=i+1}^{n/2} \mathrm{CNOT}(i,j)\right).
\end{align}
Replacing $\Uent$ by $\Uent'$ decreases the number of gates across the two sets.

Let $R \subseteq \set{1, \dotsc, D}$ denote the set of locations where we replace $\Uent$ by $\Uent'$ and let $D_1$ denote the number of remaining $\Uent$ gates. In our numerical experiments, $n = 6$ and we compare three different regimes:
\begin{itemize}
  \item $D = 9, D_1 = 9, R = \set{}$;
  \item $D = 9, D_1 = 3, R = \set{1,2,4,6,8,9}$;
  \item $D = 3, D_1 = 3, R = \set{}$.
\end{itemize}

For each run of the experiment, we generate a Hamiltonian
\begin{equation}
  H = \sum_{i=1}^{L} \alpha_i \of[\bigg]{\bigotimes_{j=1}^{n}\sigma_{i,j}},
\end{equation}
where $\alpha_i$ is drawn uniformly at random from $[-1,1]$ and $\sigma_{i,j}$ is a random Pauli matrix drawn uniformly from the set $\set{I,X,Y,Z}$, with all random variables drawn independently. We set the number of terms to be $L = 50$. After generating the Hamiltonian $H$, we execute $10^5$ steps of the simultaneous perturbation stochastic approximation (SPSA) method to determine the parameters $\theta$. At step $t$, we record the relative error of the current expectation value and the optimal eigenvalue:
\begin{equation}
  \left|\frac{\bra{0}\xp{n} U(\theta(t))\ct H U(\theta(t)) \ket{0}\xp{n} - v_{\mathrm{opt}}}{v_{\mathrm{opt}}}\right|,
\end{equation}
where $\theta(t)$ are the current parameter values and $v_{\mathrm{opt}} = \min_{\ket{\phi}}\bra{\phi}H\ket{\phi}$ is the smallest eigenvalue of $H$ obtained by explicitly diagonalizing the Hamiltonian. We run each experiment $100$ times and average the relative error over these runs. The code can be found at \url{https://github.com/TianyiPeng/Partiton_VQE}.

We report our numerical observations in \cref{fig:VQE}. Our results show that, for the problem at hand, replacing some $\Uent$ by $\Uent'$ does not greatly influence the performance of the VQE. This suggests that our scheme is worth pursuing when implementing VQE with limited quantum memory. Because of its heuristic nature, theoretical analysis of our scheme is difficult and awaits future exploration.

\begin{figure}
\includegraphics[scale=0.48]{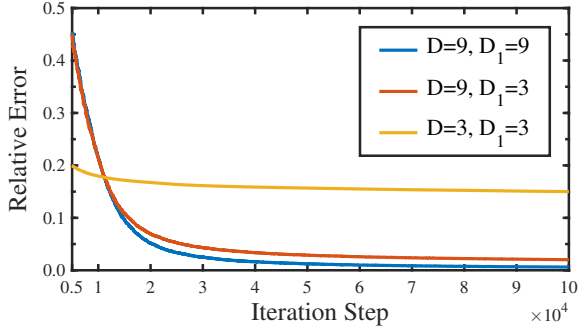}
\caption{In this example, we consider $n=6$ qubits and circuits of depth $D=9$. Let $\Uent$ be a sequence of CNOT gates. We choose $\Uent'$ by removing gates from $\Uent$ between the sets $\set{1,\dotsc,\frac{n}{2}}$ and $\set{\frac{n}{2}+1,\dotsc,n}$. We keep $D_1$ out of $D$ occurrences of $\Uent$ in $U(\theta)$ and replace the rest by $\Uent'$.
We use randomly sampled $H$ to test the performance of different parameter combinations: blue $(D=9, D_1=9)$, red $(D=9,D_1=3)$, and yellow $(D=3, D_1=3)$.
Our numerical results suggest that the behaviors of the blue and red settings are similar, i.e., replacing some $\Uent$ by $\Uent'$ may not influence the VQE performance. Both perform better than the yellow setting which has the same depth of $\Uent$ gates as the red one, but without any $\Uent'$ gates.
}
\label{fig:VQE}
\end{figure}

%
%

\end{document}